\documentclass[10pt,journal,final]{IEEEtran}
\usepackage[comma,numbers,square,sort&compress]{natbib}
\usepackage{amsmath}
\usepackage{amssymb}
\usepackage{float}
\usepackage{amsthm}
\usepackage{graphicx}
\usepackage{subfigure}
\usepackage{epstopdf}
\usepackage{color}
\usepackage{verbatim}
\usepackage{tikz,times}
\usepackage{algorithm}
\usepackage{algorithmic}

\DeclareMathOperator{\diag}{diag}

\newcommand{\norm}[1]{\left\lVert#1\right\rVert}

\newcommand{\ceil}[1]{\left\lceil #1 \right\rceil}

\def\thesection{\arabic{section}}

\theoremstyle{plain}
\newtheorem{theorem}{Theorem}
\newtheorem{proposition}{\textbf{Proposition}}
\newtheorem{lemma}{\textbf{Lemma}}

\theoremstyle{definition}
\newtheorem{definition}{\textbf{Definition}}
\newtheorem{assumption}{\textbf{Assumption}}

\begin{document}

\title{How to Secure Distributed   Filters\\ Under Sensor Attacks}

\author{Xingkang~He,  
	Xiaoqiang Ren,
	Henrik Sandberg,
	Karl H. Johansson~\IEEEmembership{Fellow,~IEEE} 
\thanks{The work is supported by Knut \& Alice Wallenberg foundation, and by Swedish Research Council.}
\thanks{X. He, H. Sandberg and K. H. Johansson are with Division of Decision and Control Systems, School of Electrical Engineering and Computer Science, KTH Royal Institute of Technology,  and they are also affiliated with Digital Futures, Sweden ((xingkang,hsan,kallej)@kth.se).}%
\thanks{	X. Ren is with  Shanghai Key Laboratory of Power Station Automation Technology, School of Mechatronic Engineering and Automation, Shanghai University, Shanghai, 200444, China.  (xqren@shu.edu.cn)}
}


\maketitle


\begin{abstract}	
We study how to secure distributed filters for linear time-invariant systems with bounded noise under false-data injection  attacks.
A malicious attacker is able to arbitrarily manipulate the observations  for   a time-varying and unknown subset of the sensors. 
We first propose a recursive distributed filter consisting of two steps at each update. The first step employs a saturation-like scheme, which gives a small gain if the innovation is  large corresponding to a potential attack. The second step is a consensus operation of state estimates among neighboring sensors.
We prove the  estimation error is upper bounded if the filter parameters satisfy a  condition. 
We further analyze the feasibility of the condition and connect it to  sparse observability in the centralized case.
When the attacked sensor set is known to be time-invariant, the secured filter is modified by adding an  online  local attack detector.
The detector is able to identify the attacked sensors whose observation innovations are larger than the detection thresholds. Also, with more attacked sensors being detected, the thresholds will adaptively adjust to reduce the space of the stealthy attack signals.
The resilience of the secured filter with detection is verified by an explicit relationship between the upper bound of the estimation error and the number of detected attacked sensors.
Moreover, for the noise-free case, we prove that the state estimate of each sensor  asymptotically converges to the system state under certain conditions. 
Numerical simulations are   provided to illustrate the developed results.
\end{abstract}

\section{Introduction}
A cyber-physical system (CPS) is a physical system controlled and monitored by computer-based algorithms.   During  recent years, numerous applications in sensor networks, vehicle networks, process control, smart grid, etc, have been  investigated.  With higher integration of large-scale computer networks and complex physical processes, these systems
are confronting more security issues both in the cyber and physical layers. 
Thus, the research   on CPS security  is attracting more and more attention.

Sensors and sensor networks are  utilized to collect  environmental data in a CPS. The quality of these sensors is essential for decision making.
However,  with the increasing number  of complex tasks and the large-scale deployment of cheap and low-quality sensors,  the vulnerability of   system operation is inevitably increased. In this paper, we consider the 
false-data injection (FDI)  attacks in sensors networks, which is  illustrated in Fig.~\ref{fig:attack_networks}, where a distributed sensor network with 30 sensors is deployed to collaboratively observe the state of a CPS. In this case,   6 sensors in red are under attack in the sense that their observations can be arbitrarily manipulated. We are interested to find a  distributed filter to estimate the system state by employing the information provided by the sensor network in Fig.~\ref{fig:attack_networks}.

\begin{figure}[t]
	\centering
	\begin{tikzpicture}[scale=0.5, transform shape,line width=1pt]
	\node [draw,shape=circle,line width=2pt,minimum size=0.7cm] (1) at (0, 0) {1};
	\node[draw,shape=circle,line width=2pt,minimum size=0.7cm] (2) at +(2*1,0) {2};
	\node[draw,shape=circle,line width=2pt,fill=red,minimum size=0.7cm] (3) at +(2*2,0) {3};
	\node[draw,shape=circle,line width=2pt,minimum size=0.7cm] (4) at +(2*3,0) {4};
	\node[draw,shape=circle,line width=2pt,minimum size=0.7cm] (5) at +(2*4,0) {5};
	\node[draw,shape=circle,line width=2pt,minimum size=0.7cm] (6) at +(2*5,0) {6};
	\node[draw,shape=circle,line width=2pt,minimum size=0.7cm] (7) at (0, -2) {7};
	\node[draw,shape=circle,line width=2pt,minimum size=0.7cm] (8) at +(2,-2) {8};
	\node[draw,shape=circle,line width=2pt,minimum size=0.7cm] (9) at +(2*2,-2) {9};
	\node[draw,shape=circle,line width=2pt,minimum size=0.7cm] (10) at +(2*3,-2) {10};
	\node[draw,shape=circle,line width=2pt,minimum size=0.7cm] (11) at +(2*4,-2) {11};
	\node[draw,shape=circle,line width=2pt,fill=red,minimum size=0.7cm] (12) at +(2*5,-2) {12};
	\node[draw,shape=circle,line width=2pt,fill=red,minimum size=0.7cm] (13) at (0, -2*2) {13};
	\node[draw,shape=circle,line width=2pt,minimum size=0.7cm] (14) at +(2,-2*2) {14};
	\node[draw,shape=circle,line width=2pt,fill=red] (15) at +(2*2,-2*2) {15};
	\node[draw,shape=circle,line width=2pt,minimum size=0.7cm] (16) at +(2*3,-2*2) {16};
	\node[draw,shape=circle,line width=2pt,minimum size=0.7cm] (17) at +(2*4,-2*2) {17};
	\node[draw,shape=circle,line width=2pt,minimum size=0.7cm] (18) at +(2*5,-2*2) {18};
	\node[draw,shape=circle,line width=2pt,minimum size=0.7cm] (19) at (0, -2*3) {19};
	\node[draw,shape=circle,line width=2pt,minimum size=0.7cm] (20) at +(2,-2*3) {20};
	\node[draw,shape=circle,line width=2pt,minimum size=0.7cm] (21) at +(2*2,-2*3) {21};
	\node[draw,shape=circle,line width=2pt,minimum size=0.7cm] (22) at +(2*3,-2*3) {22};
	\node[draw,shape=circle,line width=2pt,fill=red] (23) at +(2*4,-2*3) {23};
	\node[draw,shape=circle,line width=2pt,minimum size=0.7cm] (24) at +(2*5,-2*3) {24};
	\node[draw,shape=circle,line width=2pt,minimum size=0.7cm] (25) at (0, -2*4) {25};
	\node[draw,shape=circle,line width=2pt,minimum size=0.7cm] (26) at +(2,-2*4) {26};
	\node[draw,shape=circle,line width=2pt,minimum size=0.7cm] (27) at +(2*2,-2*4) {27};
	\node[draw,shape=circle,line width=2pt,fill=red,minimum size=0.7cm] (28) at +(2*3,-2*4) {28};
	\node[draw,shape=circle,line width=2pt,minimum size=0.7cm] (29) at +(2*4,-2*4) {29};
	\node[draw,shape=circle,line width=2pt,minimum size=0.7cm] (30) at +(2*5,-2*4) {30};
	\node[draw,shape=circle,line width=2pt,minimum size=0.7cm] (normal) at +(2*1,-2*5) {};
		\draw  (2*1.9,-2*5) node {attack-free sensor};	
	\node[draw,shape=circle,line width=2pt,fill=red,minimum size=0.7cm](attacked) at +(2*3.2,-2*5) {};
		\draw  (2*4,-2*5) node {attacked sensor};
	\foreach \from/\to in {1/2,1/7,2/8,2/9, 3/4,3/9,3/10,4/11, 5/10, 5/12,6/12,7/14,8/15,9/14,9/16,11/16,11/17,11/18,13/19,13/20,14/20,15/22,16/22,17/12,17/23,18/24,19/26,19/25,20/21,20/27,21/26,21/16,22/28,22/23,22/29,24/29,24/30,27/28}
	\draw [black] (\from) -- (\to);
	\end{tikzpicture}
	\caption{A distributed sensor network under FDI attacks.}
	\label{fig:attack_networks}
\end{figure}
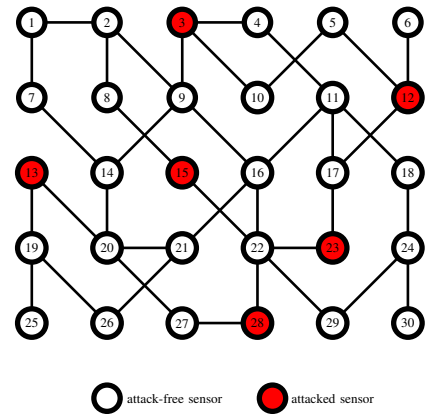

A large number of distributed filters for sensor networks have been proposed  in the  literature, e.g., \cite{Battistelli2015Consensus,liu2015event}.  These filters, however, would  not work well   in attack scenarios like Fig.~\ref{fig:attack_networks}. {In order to   degrade  filter performance, an attacker can 
strategically inject false data into the observations of   attacked sensors based on its knowledge of   systems. When     
probability distributions of  the observations are affected, the  filters prior-designed based on the  distributions are no longer effective.}
For the scenario in Fig.~\ref{fig:attack_networks}, the following questions are answered in this paper:
\begin{enumerate}
	\item How to design a  distributed filter  such that it is   resilient {when the sensor network is under FDI attacks}? 
	\item What is the maximal number of  sensors { under FDI attacks}, such that  filter  stability is  guaranteed? 
	\item How to detect which sensors are attacked  and how to  remove their influence on the filter performance?
\end{enumerate}

\subsection*{Related Work}
{The security problems of CPSs have been extensively studied in the literature
 based on  centralized frameworks, where 
  a data center   is  able to collect and process the data from all sensors. 
	 }
 To find out whether sensors are under attack and  to identify the attack signals  inserted to the systems, a study on attack detection and identification for CPSs was given in \cite{pasqualetti2013attack}, where  the design methods and analysis techniques for centralized monitors were discussed as well. 
A probabilistic approach was given in \cite{ren2019secure} to estimate a static parameter in a fusion center under   sparse FDI  sensor attacks. 
To obtain attack-resilient state estimates, some centralized state estimators or observers were proposed based on optimization techniques \cite{fawzi2014secure,pajic2017attack,pajic2017design,shoukry2017secure,shoukry2018smt,han2019convex}, which usually face heavy computational complexity in the brute force search.

{In comparison with centralized frameworks,  
	 distributed frameworks have no data center. Distributed methods rely on local computation and neighbor communication, thus they outweigh centralized methods in  scalability for large networks and robustness to failures.} 
In recent years,  some investigations in the study of sensor networks  under Byzantine   attacks/failures have been made for the distributed state estimation of dynamical systems \cite{mitra2019byzantine,mitra2019resilient},  the distributed  identification of a static vector parameter \cite{su2019finite}, and the distributed stochastic gradient descent \cite{blanchard2017machine}. 
{Although these papers studied the worst sensor attacks (i.e.,   Byzantine attacks), they   require complete connectivity or strong robustness of  graphs, which would be quite restrictive for the systems suffering milder attacks (e.g., FDI attacks).} In this direction, 
  a distributed observer with attack detection  was proposed to deal with a class of bias attacks in the observer update or sensor communication \cite{deghat2019detection}.
Distributed estimation for a static parameter under FDI sensor attacks was studied in \cite{chen2019resilient,chen2019resilienttsp}.  In \cite{an2019distributed},  a distributed optimization based method was utilized to achieve  convergence of the observer  under   sparse observability for linear time-invariant (LTI) systems \cite{chong2015observability} suffering FDI sensor attacks. Nevertheless, the results relied on the redesign of topology graph and infinite sensor communications between two observation updates. The authors in \cite{chen2017distributed} studied the distributed dimensionality reduction fusion
estimation for CPSs under   denial-of-service attacks. 
 In \cite{boem2017distributed},  for   FDI attacks  in communication networks, a distributed detection problem was studied for a group of interconnected subsystems, and an extended application to   DC microgrids was given in \cite{gallo2020distributed}.  In \cite{forti2018distributed}, a Bayesian framework based  joint distributed attack detection and state estimation were investigated in a cluster-based sensor network by considering     FDI attacks in the communication between remote sensors and fusion nodes. 
 However,  the accurate probability distribution of attacks was required. To the knowledge of the authors, there were few results considering how to achieve the co-design of a distributed estimator and an attack detector.

\subsection*{Objectives and Contributions of the Paper}
{
	In this paper, we study the distributed state estimation problem for LTI systems with bounded noise  over a sensor network, where the observations of  a time-varying and unknown subset of sensors are arbitrarily manipulated by a malicious attacker  through FDI attacks. 
	
	The objective of this paper is four-fold:	
	i) Design a resilient distributed filter for each sensor with the potentially compromised observations and the data received from neighboring sensors; ii) Analyze the main properties of the filter, including the  estimation error boundedness; iii) Design an attack detection based filter if the compromised sensor set is known to be time-invariant; iv) Analyze the main properties of the detection based filter.

Corresponding to the four objectives, this paper makes four contributions summarized in  the following:}

 i)  We design a secured distributed filter consisting of two steps (Algorithm~\ref{alg:A}).   The first step employs a saturation-like scheme, which gives a small gain if the innovation is  large corresponding to a potential attack.	
	The second step is a consensus operation of state estimates among neighboring sensors. 
	
  ii) We investigate some properties of the secured filter. First, 
	we prove that the  estimation error is upper bounded if the filter parameters satisfy a condition, whose 
	feasibility   is studied by providing an  easy-to-check sufficient and necessary condition (Theorems~\ref{thm_stability} and \ref{thm_iff}).  
	 We further connect this condition to sparse observability in the centralized case (Proposition \ref{lem_iff}).
 Moreover, we provide a  condition such that the observations of the attack-free sensors will  not be saturated after a finite time. Then,  a tighter error bound is obtained (Theorem~\ref{thm_normal}).

  iii)	{ We modify the secured distributed filter by adding an attack detector (Algorithm~\ref{alg:B}), when the   set of  attacked sensors  is known to be time-invariant.} The detector is able to identify the attacked sensors whose observation innovations are larger than the detector thresholds (Proposition \ref{lem_bound_normal}).  Moreover, with more attacked sensors being detected, the thresholds will adaptively adjust to reduce the space of the stealthy attack signals.

 iv) We study some properties of the secured filter with  attack detection. First,   
	the resilience of the  filter   is verified by an explicit relationship between the upper bound of the estimation error and the number of detected attacked sensors (Theorem~\ref{thm_bounds_detected}).
	Moreover, for the noise-free case, we prove that the state estimate of each sensor   asymptotically converges to the system state under certain conditions (Theorem~\ref{thm_observer}).

{This paper designs a filter with an innovation-dependent update gain,  essentially different from  conventional filters with statistics-based gains (e.g., Kalman filter), in order to confine the influence of attack signals to the  estimation.} To handle the technical  difficulties in performance analysis,  a new  tool inspired by bounded-input  bounded-output (BIBO) stability is provided to analyze boundedness of the estimation error.
The distribution assumption on  attack signals in \cite{deghat2019detection} is removed in this work by allowing that the attacker  can inject any attack signals. Moreover, the assumption that the attacked sensor set is fixed over time in both centralized frameworks \cite{pajic2017attack,fawzi2014secure,shoukry2017secure,han2019convex,ren2019secure,nakahira2018attack,shoukry2018smt} and distributed frameworks \cite{mitra2019byzantine,mitra2019resilient,an2019distributed,lee2020fully} is extended to the time-varying case.
{The robustness requirement of communication graphs in \cite{mitra2019byzantine,mitra2019resilient} for a wider range of attacks and  the requirement of infinite communication rate between two updates  in \cite{an2019distributed} are both removed in this paper.}
This paper builds on  the  preliminary work   presented in \cite{he2019securecdc} and
\cite{He2020Secured}. The main difference is four-fold. First, the set of the
attacked sensors is extended from the time-invariant case to the
time-varying case. Second, a new section dealing with attack
detection and sensor isolation is added. Third, the results in \cite{he2019securecdc} and
\cite{He2020Secured} are generalized  and new theoretical
results with proofs are added. Fourth, more
literature comparisons and simulation results are provided.

The remainder of the paper is organized as follows: Section~\ref{sec_formulation} is on the problem formulation. Section \ref{sec_filter} provides the secured distributed filter and its performance analysis.  The secured distributed  filter with an online attack detector is studied in Section \ref{sec:detector}. 
After  numerical simulations in Section~\ref{sec_simu}, the paper is concluded in   Section \ref{sec_conclusion}.
The main proofs are given in Appendix.

\emph{Notations.}  $\mathbb{R}^n$ is the set of $n$-dimensional real vectors. $\mathbb{R}^+$ and $\mathbb{Z}^+$ are the sets of  positive real scalars and integers, respectively. $\mathbb{R}^{n\times m}$ is the set of real matrices with $n$ rows and $m$ columns. 
$\diag\{\cdot\}$   represents the diagonalization operator.  
$I_{n}$ stands for the $n$-dimensional square identity matrix. 
$\textbf{1}_N$ stands for the $N$-dimensional vector with all elements being one. 
{The superscript ``$\sf T$" represents the transpose. }
$A\otimes B$ is the Kronecker product of $A$ and $B$.  $\norm{x}$ is the 2-norm of a vector $x$. $\norm{A}$ is the induced 2-norm of matrix $A$, i.e., $\norm{A}=\sup\limits_{x\neq 0}\frac{\norm{Ax}}{\norm{x}}$.  
$\lambda_{\min}(A)$, $\lambda_2(A)$ and $\lambda_{\max}(A)$ are the minimum, second minimum and maximum eigenvalues of a real-valued symmetric matrix $A$, respectively. 
$|\Gamma|$ is the cardinality of the set $\Gamma.$ 
$\min\{a,b\}$ means the minimum between the real-valued scalars $a$ and $b$. For a set $\mathcal{A}$, the indicator function $\mathbb I_{a\in \mathcal{A}}=1$, if $a\in \mathcal{A}$; $\mathbb I_{a\in \mathcal{A}}=0$, otherwise. $\ceil{\cdot}$ is the ceiling function.

\section{Problem Formulation}\label{sec_formulation}
In this section, we first provide some graph preliminaries and then set up the problem of this paper.
\subsection{Graph Preliminaries}
We model the communication topology of $N$ sensors by an undirected graph $\mathcal{G=(V,E)}$ without self loops, where $\mathcal{V}=\{1,2,\dots,N\}$ stands for the set of  nodes, and $\mathcal{E}\subseteq \mathcal{V}\times \mathcal{V}$ is the  set of edges.  If there is an edge $(j,i)\in \mathcal{E}$, node $i$ can exchange information with node $j$, then node $j$  is called a  neighbor of node $i$, and vice versa. 
Let the neighbor set of node $i$ be $\mathcal{N}_{i}:=\{j\in\mathcal{V}|(j,i)\in \mathcal{E}\}$. The degree matrix of $\mathcal{G}$ is $D_{\mathcal{G}}=\diag\{|\mathcal{N}_{1}|,\dots,|\mathcal{N}_{N}|\}$. The adjacency matrix is $\mathbb{A}_{\mathcal{G}}=[a_{i,j}]$, where  $a_{i,j}=1$ if $(i,j)\in \mathcal{E}$,  otherwise $a_{i,j}=0$. $\mathcal{L}=D_{\mathcal{G}}-\mathbb{A}_{\mathcal{G}}$ is the Laplacian matrix. 
Graph $\mathcal{G}$ is  connected if for any pair of two different nodes $i_{1},i_{l}$, there exists a  path from $i_{1}$ to $i_{l}$ consisting of edges $(i_{1},i_{2}),(i_{2},i_{3}),\ldots,(i_{l-1},i_{l})$.   
On the connectivity of a graph, we have the following proposition.
\begin{proposition}\cite{Mesbahi2010Graph}\label{thm_graph}
	The undirected graph $\mathcal{G}$ is	 connected if and only if $\lambda_{2}(\mathcal{L})>0$.
\end{proposition}

\subsection{System Model}
For a sensor network $\mathcal{G}$ under {FDI attacks}, we illustrate the scenario in Fig.~\ref{fig:attack_networks}, where each senor is equipped with a filter to estimate the system state (see Fig.~\ref{fig:diag} for a diagram). The state-space system model is given as follows
\begin{equation}\label{eq_system}
\begin{split}
x(t+1)&=Ax(t)+w(t)\\
y_i(t)&=C_ix(t)+v_i(t)+a_i(t),i=1,\dots,N,
\end{split}
\end{equation}
where $x(t)\in\mathbb{R}^n$ is the unknown system state, $w(t)\in\mathbb{R}^n$ the process  noise,  $v_i(t)\in\mathbb{R}$  the observation  noise, $a_i(t)\in\mathbb{R}$ the  attack signal inserted by some malicious attacker, and $y_i(t)\in\mathbb{R}$  the observation of sensor $i$, all at time $t$.  Moreover, $A\in\mathbb{R}^{n\times n}$ is the system state transition matrix, and $C_i\in\mathbb{R}^{1\times n}$ is the observation vector of sensor $i$. {Both $A$ and $C_i$ are known to each sensor. We do not assume that $(A,C_i)$ is   observable.} Without losing generality, we assume that the observation vectors are normalized, i.e.,  $\norm{C_{i}}= 1, i\in\mathcal{V}=\{1,\dots,N\}$. Otherwise, we can reconstruct the observation equation of system \eqref{eq_system}.

In this paper, we consider the observation equation in \eqref{eq_system} with scalar outputs for each sensor. This conforms with the centralized framework \cite{fawzi2014secure}, where each row vector of the centralized observation matrix stands for the observation vector of one sensor. For the case that  outputs of some sensors are not scalar, we can replace each of these sensors by a set of virtual sensors with scalar ouputs, which are completely connected and connected to the neighbors of the original sensor.
 Then, the problem will reduce to the one studied in this paper.

\begin{figure}[t]
	\centering
	\begin{tikzpicture}[scale=0.8, transform shape,line width=1pt]
	\draw (0,0)  node[rectangle,draw,scale=1.5,line width=2pt] (plant)   {Plant};
\draw (0,1)  node (w)   {$w$};
\draw (2,1.5)  node (v)   {$v_i$};
\draw (4,-1.5)  node (a)   {$a_i$};
\draw (6,1.5)  node (x)   {$\{\hat x_j\}_{j\in\mathcal{N}_i}$};
\draw (2,0)  node[circle,draw] (plus1)   {+};	
\draw (4,0)  node[circle,draw] (plus2)   {+};
\draw (6,0)  node[rectangle,draw,scale=1.5,line width=2pt] (filter)   {Filter};
\draw (8,0)  node  (output)   {};
\draw[draw=black,dashed] (1,-0.7) rectangle ++(6,1.5);  
\draw[->] (w) -- (plant);
\draw[->] (v) -- (plus1);
\draw[->] (a) -- (plus2);
\draw[->] (x) -- (filter);
\draw[->] (plant) -- (plus1);
\draw[->] (plus1) -- (plus2);
\path [->] (plus2) edge node[above] {$y_i$} (filter);
\path [->] (filter) edge node[above] {$\hat x_i$} (output);
	\end{tikzpicture}
	\caption{Each sensor is equipped with a filter providing an estimate $\hat x_i$ of state $x$. The sensor observation $y_i$ is potentially compromised through an attack signal $a_i$.}
	\label{fig:diag}
\end{figure}
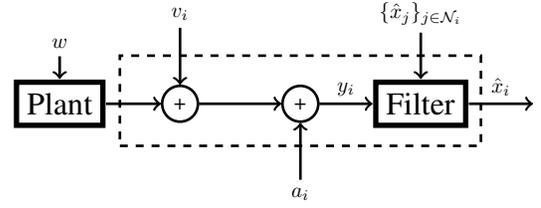

The following assumptions are  needed.
\begin{assumption}\label{ass_bounds}
	The following conditions hold
	\begin{align*}
	&\sup_{t}\norm{w(t)}\leq b_{w}\\
	&\max_{i\in\mathcal{V}}\sup_{t}\norm{v_i(t)}\leq b_{v}\\
	&\norm{\hat x(0)-x(0)}\leq \eta_0, 
	\end{align*}
	where $\hat x(0)$ is the   estimate of $x(0)$ shared by all sensors, {and the upper bounds are known to each sensor.}
\end{assumption}

\begin{assumption}\label{ass_graph}
	Communication graph $\mathcal{G}$ is   connected.
\end{assumption}

\subsection{Attack Model}

{To deteriorate  estimation performance, a malicious attacker can  compromise the observations of some targeted sensors by FDI attacks.} 
However, due to  resource limitation, the attacker can only attack a subset of all sensors at each time. 
Let $\mathcal{A}(t)$ and $\mathcal{A}^c(t)$ be the  set  of attacked sensors and the set of attack-free sensors  at time $t$, respectively. It holds that $|\mathcal{A}(t)|+|\mathcal{A}^c(t)|=N.$ We require the following assumption on the attack model.
\begin{assumption}\label{ass_attacker}
The attacker can implement the following   {FDI} attacks to system \eqref{eq_system}: for $t=1,2,\dots,$
\begin{equation}\label{attack_input}
\begin{split}
&a_i(t)\in \mathbb{R}, i\in\mathcal{A}(t), |\mathcal{A}(t)|\leq s \\
&a_i(t)=0, i\in \mathcal{A}^c(t),
\end{split}
\end{equation}
where   sets $\mathcal{A}(t)$ and $\mathcal{A}^c(t)$ are unknown to each sensor, but  $s$ is known.
\end{assumption}

	In Assumption~\ref{ass_attacker}, we consider the worst scenario on { FDI attacks} that the attacker can inject attack signals with any distribution, which is more general than  results in the  literature \cite{deghat2019detection}. Moreover, Assumption~\ref{ass_attacker} removes the requirement in \cite{pajic2017attack,fawzi2014secure,shoukry2017secure,han2019convex,ren2019secure,nakahira2018attack,shoukry2018smt,mitra2018secure} that the attacked sensor set is fixed over time.

\subsection{Problem of Interest}

{Design a resilient distributed filter $\{\hat x_i(t)\}_{i\in\mathcal{V}}$ for   system \eqref{eq_system} under Assumptions~\ref{ass_bounds}--\ref{ass_attacker}  by employing   potentially compromised sensor observations $\{y_{i}(l)\}_{l=1}^t$  and the received neighbor messages  over   communication graph $\mathcal{G}$, such that 
	\begin{align*}
\limsup_{t\rightarrow\infty}\|\hat x_i(t)-x(t)\|\leq \Delta,
\end{align*}
where $\Delta\geq 0$    reflects the performance of the proposed filter. 
Moreover, find the answers to the questions 1)--3)   in the introduction.

}

\section{Secured Distributed  Filter}\label{sec_filter}
In this section, we first design  a secured distributed filter for each sensor and then analyze  some properties of the filter.
\subsection{Filter Design}\label{subsec:filter}
We consider the filter with two steps, namely, observation update and estimate consensus. In the step of  observation update, by choosing $\beta>0$,  we design a saturation-like scheme to utilize   observation $y_{i}(t)$ as follows
\begin{align}\label{alg_update}
\tilde x_{i}(t)=&A\hat x_{i}(t-1)+ k_{i}(t)C_i^{\sf T}(y_{i}(t)-C_iA\hat x_{i}(t-1)), 
\end{align}
where 
\begin{align}\label{eq_K}
k_{i}(t)=\begin{cases}
1,\text{ if } |y_{i}(t)-C_{i}A\hat x_{i}(t-1)|\leq\beta,\\
\frac{\beta}{|y_{i}(t)-C_{i}A\hat x_{i}(t-1)|}, \text{ otherwise}.
\end{cases}
\end{align}
Different from the gains  of conventional filters or state observers (e.g., Kalman filter),   gain $k_{i}(t)$  is related to the estimation innovation (i.e., $y_{i}(t)-C_iA\hat x_{i}(t-1)$). 
The design of $k_{i}(t)$ in (\ref{eq_K}) makes sense, since if the  innovation  is   large,   observation $y_i(t)$ is more likely to be compromised. 
Note that $|k_{i}(t)(y_{i}(t)-C_iA\hat x_{i}(t-1))|\leq\beta$, which ensures that the attacker has limited influence to the local update. 
 Scalar $\beta$  is an observation confidence parameter reflecting the usage tradeoff between attack-free observations and attacked observations.	
 If $\beta$ is very large,   almost all attack-free observations will be utilized without saturation. However, it will give much   space for the attacker  to deteriorate the estimation performance. If $\beta$ is very small, although  most  attack signals $\{a_i(t)\}$ may be filtered,  attack-free observations will contribute little to the estimation.
 The design of  $\beta$  will be discussed in the next subsection. 

In the step of estimate consensus, {we consider a two-time-scale scheme with communication rate $L\geq 1$, i.e.,  each sensor can communicate with its neighbors for $L\geq 1$ times between two measurement updates.}
%
For $l=1,2,\dots,L,$ and $\alpha>0,$
\begin{align}\label{alg_consensus}
\hat x_{i,l}(t)=\hat x_{i,l-1}(t)-\alpha\sum_{j\in\mathcal{N}_{i}}(\hat x_{i,l-1}(t)-\hat x_{j,l-1}(t)),
\end{align}
with $\hat x_{i,0}(t)=\tilde x_{i}(t)$ and  $\hat x_{i}(t)=\hat x_{i,L}(t).$ In the $l$-th communication, sensor $j$  transmits its estimate $\hat x_{j,l-1}(t)$ to its neighbors, $l=1,\dots,L.$
	The term $\alpha\sum_{j\in\mathcal{N}_{i}}(\hat x_{i,l-1}(t)-\hat x_{j,l-1}(t))$ is to make sensor estimates  tend to consensus.
	Communication rate $L$ is vital to guarantee the bounded estimation error   especially for the case that each subsystem is not observable (i.e., $(A,C_i)$ is not observable). 
	It can be proven that if   communicate rate $L$ goes to infinity and   parameter $\alpha$ is properly designed,    estimates $\{\hat x_{i}(t)\}_{i=1}^N$ will converge to the same vector.
	However,  an infinite communication  rate  in \cite{zhao2017resilient,chen2019resilient}   is not necessary in this work. The  design of $L$ and $\alpha$ is studied in the next subsection. 
By (\ref{alg_update})--(\ref{alg_consensus}), we provide the   {distributed saturation-based  filter}  in Algorithm~\ref{alg:A}.
\begin{algorithm}[t]
	\caption{{Distributed  Saturation-Based  Filter}}
	\label{alg:A}
	\begin{algorithmic}[1]
				\STATE {\textbf{Initial setting:} ($\hat x_{i}(0),\alpha,\beta,L$)}\\		\vskip 2pt
				\FOR{$t=1,2,\dots$}
		\STATE {\textbf{Observation update:}}\\		\vskip 2pt
$k_{i}(t)=\min\{1,\frac{\beta}{|y_{i}(t)-C_{i}A\hat x_{i}(t-1)|}\}$\\\vskip 2pt
		$\tilde x_{i}(t)=A\hat x_{i}(t-1)
		+ k_{i}(t)C_i^{\sf T}(y_{i}(t)-C_iA\hat x_{i}(t-1))\nonumber$\\   		\vskip 2pt     
		\STATE {\textbf{Estimate consensus}: Let $\hat x_{i,0}(t)=\tilde x_{i}(t)$}\\		\vskip 2pt
		\FOR{$l=1,\dots,L$}
		\STATE{
	  Sensor $i$ receives $\hat x_{j,l-1}(t)$ from neighbor sensor $j$,  \\		\vskip 2pt
		$ \hat x_{i,l}(t)=\hat x_{i,l-1}(t)-\alpha\sum_{j\in\mathcal{N}_{i}}(\hat x_{i,l-1}(t)-\hat x_{j,l-1}(t))$\\		\vskip 2pt	
	}
        \ENDFOR
		\STATE {Let $\hat x_{i}(t)=\hat x_{i,L}(t).$ }
		\ENDFOR
	\end{algorithmic}
\end{algorithm}
\subsection{Performance Analysis}\label{sec_analysis}
	Since   filtering gains $\{k_i(t)\}_{i=1}^N$ in \eqref{eq_K} are related to the state estimates and  potential compromised observations, the common stability analysis approaches, such as Lyapunov methods, may not be directly utilized.
This is the main technique challenge of this paper. 
Inspired by BIBO stability, 
 we provide the following lemma 
to analyze boundedness of the estimation error.

\begin{lemma}\label{lem_stability}
	Consider a one-dimensional equation $x_{t+1}=F(x_t)x_t+q_0$ at time $t\geq 0$, where $x_0\geq 0$,   $q_0\geq 0,$ and $F(\cdot)\in[0,1]$ is a monotonically non-decreasing function. If   set $\Gamma=\{t\geq 1|x_{t}\leq x_{t-1}\}$ is non-empty, the following conclusions hold:
	\begin{enumerate}
				\item If $q_0\neq 0,$   for $\forall t_0\in \Gamma,$
				 $$x_t\leq F^{t-t_0}(x_{ t_0})x_{ t_0}+q_0\frac{1-F^{t-t_0}(x_{ t_0})}{1-F(x_{ t_0})},t\geq t_0;$$\vskip 5pt
		\item $\sup\limits_{t\geq  t_0}x_t\leq x_{ t_0}$, $\forall t_0\in \Gamma;$\vskip 5pt

		\item $\limsup\limits_{t\rightarrow \infty}x_t\leq \inf\limits_{t_0\in \Gamma} x_{ t_0}.$
	\end{enumerate}	

\end{lemma}
\begin{proof}
	See Appendix \ref{pf_lem_stability}.
\end{proof}
If we treat $x_t$ as an upper bound of the norm of the  estimation error, based on the knowledge of  $x_{t_0}$, we are able to use 1) to obtain a real-time upper bound of $x_t$, and apply 2) and 3) to obtain the uniform and asymptotic   bounds of $x_t$, respectively.  
To proceed, denote 
\begin{align}\label{eq_lambda}
\lambda_0:=\min\limits_{\mathcal{J}\subset\{1,2,\dots,N\}:|\mathcal{J}|=N-s}\lambda_{\min}\left( \sum_{i\in\mathcal{J}}C_{i}^{\sf T}C_{i}\right),
\end{align}
  where $s$ is the upper bound of the attacked sensor number, given in \eqref{attack_input}. 
  {
  		Since  $\sum_{i\in\mathcal{J}}C_{i}^TC_{i}$ is positive semi-definite and $s\leq N$, we have $\lambda_0\geq 0.$
  		Moreover, it holds that $\lambda_0\leq (N-s)\lambda_{\max}\left(C_{i}^TC_{i}\right)=N-s$,
  		where the  equality is   due to $\lambda_{\max}\left(C_{i}^TC_{i}\right)=\norm{C_i}^2=1$ assumed after   system   \eqref{eq_system}. Thus, $\lambda_0$ belongs to $[0, N-s]$.
  	}
To apply Lemma~\ref{lem_stability}, we construct    sequence $\{\rho_{t}\in\mathbb{R}|\rho_{t}\}$ in the following
\begin{align}\label{sequence}
\rho_{t+1}=F(\rho_{t})\rho_{t}+q_0,\qquad \rho_0=\eta_0
\end{align}
where $\eta_0$ is given in Assumption~\ref{ass_bounds} and
\begin{align}\label{notation}
\begin{split}
F(\rho_t)&=\norm{A}\left(1-\frac{k^*(\rho_t)}{N}\lambda_0\right),\\
k^*(\rho_t)&=\min\bigg\{1,\frac{\beta}{\norm{A}(p_0+\rho_t)+b_{w}+b_{v}}\bigg\},\\
q_0&=\frac{N-s}{N}(b_{w}+b_{v}+\norm{A}p_0)+b_{w}+\frac{s\beta}{N},\\
p_0&=\frac{\sqrt{N}\beta\gamma^L}{1-\norm{A}\gamma^L},\\
\gamma&=\frac{\lambda_{\max}(\mathcal{L})-\lambda_2(\mathcal{L})}{\lambda_{\max}(\mathcal{L})+\lambda_2(\mathcal{L})}.
\end{split}
\end{align}
{Under Assumption~\ref{ass_graph}, we have $\gamma\in[0,1)$. Define   $\gamma^{-1}=+\infty$ if $\gamma=0.$}
The following theorem studies boundedness of the  estimation error of Algorithm~\ref{alg:A}.
\begin{theorem}\label{thm_stability}
	(\textbf{Bounds})  Under Assumptions~\ref{ass_bounds}--\ref{ass_attacker}, consider Algorithm~\ref{alg:A} with $ \alpha=\frac{2}{\lambda_2(\mathcal{L})+\lambda_{\max}(\mathcal{L})}$.
	If there exist  $L>\ln \norm{A}/\ln \gamma^{-1}$, $\beta,\eta_0>0$, such that
	\begin{align}\label{condition_thm}
\eta_0(1-F(\eta_0))\geq q_0,
	\end{align}
	  set $\Gamma=\{t\geq 1|\rho_{t}\leq \rho_{t-1}\}$  is non-empty with $1\in \Gamma$,
	where    sequence $\{\rho_t\}$ is in \eqref{sequence}. Furthermore, for $i\in\mathcal{V}$,   estimation error $e_{i}(t)=\hat x_i(t)-x(t)$ satisfies the following properties:
	\begin{enumerate} 
			\item The estimation error is bounded at each time, i.e., 
			$\forall t_0\in \Gamma,$ $t\geq t_0$
			$$\norm{e_{i}(t)}\leq R(F(\rho_{ t_0}),t)+p(t),$$
				where
				\begin{align}\label{eq_p_R}
				\begin{split}
R(x,t)&=x^{t-t_0}\rho_{ t_0}+q_0\frac{1-x^{t-t_0}}{1-x},\\
p(t)&=\sqrt{N}\beta\gamma^L\frac{1-(\norm{A}\gamma^{L})^{t}}{1-\norm{A}\gamma^L}.
				\end{split}
				\end{align}
			\item The estimation error has a finite uniform upper bound, i.e., $\forall t_0\in \Gamma,$ $$\sup\limits_{t\geq  t_0}\norm{e_{i}(t)}\leq \rho_{ t_0}+\sup_{t\geq t_0}p(t).$$
			\item The   estimation error is asymptotically upper bounded, i.e.,
			$$\limsup\limits_{t\rightarrow \infty}\norm{e_{i}(t)}\leq \inf\limits_{t_0\in \Gamma} \rho_{ t_0}+\frac{\sqrt{N}\beta\gamma^L}{1-\norm{A}\gamma^L}.$$
	\end{enumerate}
\end{theorem}
\begin{proof}
See Appendix \ref{app_A}.
	\end{proof}

	{If   $t_0=1$, the bounds in Theorem~\ref{thm_stability}     directly depend on the initial condition. With the increase of $t_0$, the bounds become tighter.} The system designer with global system knowledge is able to examine condition \eqref{condition_thm} and calculate the error bounds in Theorem~\ref{thm_stability}.

In the following theorem, we   show that it is feasible to design   parameters $L>\frac{\ln \norm{A}}{\ln \gamma^{-1}}$, $\beta,\eta_0$ such that condition \eqref{condition_thm} is satisfied {under the case that the system is either marginally stable or unstable, i.e.,  $\norm{A}\geq 1$. } 
\begin{theorem}\label{thm_iff}
	(\textbf{Feasibility}) 
	It is feasible to find positive parameters $\beta,\eta_0$ and $L>\frac{\ln \norm{A}}{\ln \gamma^{-1}}$, and a scalar $\epsilon>0$ such that condition \eqref{condition_thm} holds for $\norm{A}\in [1,1+\epsilon)$,  
	if and only if
	\begin{align}\label{eq_iff}
	\lambda_0>s,
	\end{align}
	 where  $s$ and $\lambda_0$  are given in   \eqref{attack_input}  and \eqref{eq_lambda}, respectively.
\end{theorem}
\begin{proof}
See  Appendix \ref{app_thm_iff}.
	\end{proof}

{ 	
	Condition \eqref{eq_iff}    means that   the maximal number of attacked sensors at each time, i.e., $s$, is less than   scalar $\lambda_0$ which depends on the observation matrices of  attack-free sensors as shown in \eqref{eq_lambda}.
   Given $s$, it is straightforward to check \eqref{eq_iff} with the knowledge of the system observation matrices $\{C_i\}_{i=1}^N.$ 
	 For a particular system with $C_i=1$, $i=1,\dots,N,$  we have $	\lambda_0=N-s$.   Hence,   \eqref{eq_iff} is equivalent to $s\leq \lceil{ N/2\rceil}-1$, which is the same maximum obtained  
	under FDI  sensor attacks in  \cite{chong2015observability,fawzi2014secure}.  }

In the following theorem, by adding another  condition, we show   all the observations of attack-free sensors will eventually not be saturated, which contributes to   tighter bounds than those in Theorem~\ref{thm_stability}.
	\begin{theorem}\label{thm_normal}
(\textbf{Bounds}) Under the same conditions as in Theorem~\ref{thm_stability},   if there is a time $t_0\in\Gamma$ (e.g., $t_0=1$), such that 
\begin{align}\label{extra_condition}
\rho_{t_0}+\sup_{t\geq t_0}p(t)< \frac{\beta-b_w-b_v}{\norm{A}}
\end{align} 
  the following results hold:
\begin{enumerate}
	\item All the observations of attack-free sensors will eventually not be saturated, i.e., $k_i(t)=1$, $\forall i\in\mathcal{A}^c(t)$, $\forall t> t_0$;
	\vskip 5pt
	
	\item Compared to 1) of Theorem~\ref{thm_stability}, a tighter upper bound of the estimation error is ensured, i.e.,   $		\norm{e_{i}(t)}\leq R(\varpi,t)+p(t),\forall t>t_0$;
		\vskip 5pt
			\item  Compared to 3) of Theorem~\ref{thm_stability}, a tighter asymptotic upper bound of the    estimation error is ensured, i.e., 	$	\limsup\limits_{t\rightarrow \infty}\norm{e_{i}(t)}\leq \frac{q_0}{1-\varpi}+\frac{\sqrt{N}\beta\gamma^L}{1-\norm{A}\gamma^L}<\infty,$
\end{enumerate}
where $\rho_{t}$ is in \eqref{sequence}, $	p(t)$ and $ R(\cdot,t)$ are given in  \eqref{eq_p_R}, 
  and 
\begin{align}
\varpi&=\max\limits_{\mathcal{M}\subset\{1,2,\dots,N\}:|\mathcal{M}|=N-s}\norm{\left(I_n-\frac{1}{N}\sum_{i\in\mathcal{M}}C_{i}^{\sf T}C_{i}\right)A}.\nonumber
\end{align}
	\end{theorem}
\begin{proof}
See  Appendix \ref{app_thm_normal}.
\end{proof}

The main idea to design  $\beta$ is to minimize two asymptotic upper bounds in Theorems~\ref{thm_stability} and \ref{thm_normal} w.r.t. $\beta$ under the constraints in \eqref{condition_thm} and \eqref{extra_condition}, respectively. 
Since in 3) of Theorem~\ref{thm_stability}, $\inf_{t_0\in \Gamma} \rho_{ t_0}$ is not analytical w.r.t. $\beta,$ we may choose an upper bound of $\inf_{t_0\in \Gamma} \rho_{ t_0}$, e.g., $\lim_{t\rightarrow \infty}R(F(\eta_0),t)$, as the optimization loss function w.r.t. $\beta$.
Note that the two optimization problems for the two cases in Theorems~\ref{thm_stability}~and~\ref{thm_normal}  are  non-convex, where some heuristic optimization methods can be utilized.

Algorithm resilience  is on the relationship between the number of the attacked sensors and the estimation performance. 
Recall that $s$ is an upper bound of the attacked sensor number given in Assumption~\ref{ass_attacker}, then we study the relationship between $s$ and an upper bound of the estimation error  in the following proposition.
\begin{proposition}\label{prop_resi}
	Under the same conditions as in Theorem~\ref{thm_stability}, for each sensor $i\in\mathcal{V}$, the estimation error is asymptotically upper bounded, i.e.,
		\begin{align}\label{eq_bound_mono}
		\limsup\limits_{t\rightarrow \infty}\norm{e_{i}(t)}\leq f(s),
		\end{align}
where $f(s)=\frac{\bar q_0}{1-F(\eta_0)}+\frac{\sqrt{N}\beta\gamma^L}{1-\norm{A}\gamma^L}$	
 is a monotonically non-decreasing function w.r.t. $s$, in which
$\bar q_0=b_{w}+\max\{\beta,b_{w}+b_{v}+\norm{A}p_0\}$ if $\norm{A}<1$, otherwise $\bar q_0=q_0$.
\end{proposition}
\begin{proof}
	See Appendix \ref{pf_prop_resi}.
\end{proof}

\subsection{Connection with Sparse Observability}\label{subsec:observa}
{ The sparse observability  in the following  can be used in state estimation under FDI sensor attacks.}
\begin{definition}\label{def_sparse}
		The linear system defined by (\ref{eq_system})  is said to be $s$-sparse observable if for every set $\Gamma\subseteq \{1,\dots,N\}$ with $|\Gamma|=s$, the pair $(A,C_{\bar \Gamma})$ is observable, where $C_{\bar \Gamma}$ is the remaining matrix by removing $C_{j},j\in\Gamma$ from $[C_1^{\sf T},C_2^{\sf T},\dots,C_N^{\sf T}]^{\sf T}$. Furthermore, if the pair $C_{\bar \Gamma}^{\sf T}C_{\bar \Gamma}=\sum_{i=1,i\notin\Gamma}^{N}C_{i}^{\sf T}C_{i}\succ0$, the system is said to be one-step $s$-sparse observable.
\end{definition}

 In   centralized frameworks, if the observations of $s$ sensors are compromised, the system should be $2s$-sparse observable to guarantee the effective estimation of system state \cite{shoukry2016event}. 
The direct relationship between \eqref{eq_iff} and the one-step $s$-sparse observability is given in the following.
\begin{proposition}\cite{chen2019resilienttsp}\label{lem_iff}
	A necessary condition to guarantee $\lambda_0>s$ is that  system (\ref{eq_system}) is one-step $2s$-sparse observable. If the observation vectors are orthogonal,  one-step $2s$-sparse observability is also a sufficient condition to guarantee $\lambda_0>s$.
\end{proposition}

 If the observations of any $s$ sensors are compromised, based on the results of Theorems~\ref{thm_stability}, \ref{thm_iff} and Proposition \ref{lem_iff},  Algorithm~\ref{alg:A} is able to 
to achieve  effective estimation  if 
 system (\ref{eq_system}) is one-step $2s$-sparse observable.

\section{Secured distributed filter with attack detection}\label{sec:detector}
{In this section, we modify Algorithm~\ref{alg:A} by adding an attack detection scheme and then analyze the properties of the revised algorithm.}
In the sequel, we need  the following assumption, which is common in the literature \cite{pajic2017attack,fawzi2014secure,shoukry2017secure,han2019convex,ren2019secure,nakahira2018attack,shoukry2018smt,mitra2019byzantine,mitra2019resilient,an2019distributed}. 
{\begin{assumption}\label{ass_number}
		The attacked sensor set is known to be time-invariant, i.e., $\mathcal{A}(t)\equiv\mathcal{A}$, such that $|\mathcal{A}|\leq s$.
\end{assumption}}
Under Assumption~\ref{ass_number}, we denote $\mathcal{A}^c$ the complement of $\mathcal{A}$ in the sensor set $\mathcal{V}$, i.e., $\mathcal{A}^c\bigcup \mathcal{A}=\mathcal{V}$.

\subsection{Detection Based Distributed Filter}
{Denote   $\mathcal{I}_{i}(t)$ the index set of detected attacked sensors  known to sensor $i$ at time $t$. }
Moreover, we let $d_i(t)=|\mathcal{I}_{i}(t)|$. 
In order to improve   estimation performance, we aim to  isolate the observations of   sensors in the set $\mathcal{I}_{i}(t)$  in the following way:
 If   sensor $i$, $i\in\mathcal{V}$, is detected to be under attack at a time $T^*$, i.e., $i\in\mathcal{I}_{i}(T^*)$,   sensor $i$ will no longer use its observations for $t\geq T^*$.
We define  the   following sequence $\{\bar \rho_{t,i}\in\mathbb{R}\}$. 
For each $i\in\mathcal{V}$, $t=0,1,\dots,$ let 
\begin{align}\label{sequence3}
\bar\rho_{t+1,i}=\bar F(\bar\rho_{t,i},t)\bar\rho_{t,i}+\bar q_i(t),\qquad \bar\rho_{0,i}=\eta_0
\end{align}
where 
\begin{align*}
\begin{split}
\bar F(\bar\rho_{t,i},t)&=\norm{A}\left(1-\frac{k^*(\bar\rho_{t,i},t)}{N}\lambda_0\right),\\
k^*(\bar\rho_{t,i},t)&=\min\bigg\{1,\frac{\beta}{\norm{A}(p(t)+\bar\rho_{t,i})+b_{w}+b_{v}}\bigg\},\\
\bar q_i(t)&=q_0-\frac{d_i(t)\beta}{N},
\end{split}
\end{align*}
and $q_0$ is given in \eqref{notation}.
Based on     sequence $\{\bar \rho_{t,i}\in\mathbb{R}\}$,  we provide an attack detection condition in the following.

\textbf{Detection condition: }  
 Sensor $i \in\mathcal{V}$ is believed   under attack if  
\begin{align}\label{detector2}
|y_{i}(t)-C_{i}A\hat x_{i}(t-1)|>\bar\varphi_i(t),
\end{align}
where $\bar\varphi_i(t)=\norm{A}(\bar \rho_{t-1,i}+p(t-1))+b_{w}+b_{v}$, and $\bar \rho_{t-1,i}$ is given in \eqref{sequence3}.

Then, we propose 
a distributed  {saturation-based filter} with  detection in Algorithm~\ref{alg:B}, which is modified from Algorithm~\ref{alg:A} by employing the   detection condition  and the  observation isolation operation.

\begin{algorithm}[t]
	\caption{  Distributed  Saturation-Based Filter with  Detection}
	\label{alg:B}
	\begin{algorithmic}[1]
		\STATE {\textbf{Initial setting:} ($\hat x_{i}(0),\mathcal{I}_{i}(0),\alpha,\beta,L,s$)}\\		\vskip 2pt
						\FOR{$t=1,2,\dots$}
		\STATE {\textbf{Update with detection:}} Let $\mathcal{I}_{i}(t)=\mathcal{I}_{i}(t-1)$ \\ 		\vskip 2pt
		\IF{ $i\in\mathcal{I}_{i}(t)$}  
		\STATE{ $\tilde x_{i}(t)=A\hat x_{i}(t-1)$	} 
		\ELSIF{$|\mathcal{I}_{i}(t)|=:d_i(t)=s$} 
		 \STATE{$\tilde x_{i}(t)=A\hat x_{i}(t-1)
		 	+ C_i^{\sf T}(y_{i}(t)-C_iA\hat x_{i}(t-1))\nonumber$}  
		\ELSIF{$|y_{i}(t)-C_{i}A\hat x_{i}(t-1)|>\bar\varphi_i(t)$, where $\bar\varphi_i(t)$ is in \eqref{detector2}} 
		 \STATE{$\tilde x_{i}(t)=A\hat x_{i}(t-1)$,  $\text{let }\mathcal{I}_{i}(t)=\mathcal{I}_{i}(t)\cup \{i\}$}		
		\ELSE
		 \STATE{$k_{i}(t)=\min\bigg\{1,\frac{\beta}{|y_{i}(t)-C_{i}A\hat x_{i}(t-1)|}\bigg\}$\\
		 	$\tilde x_{i}(t)=A\hat x_{i}(t-1)
		 	+ k_{i}(t)C_i^{\sf T}(y_{i}(t)-C_iA\hat x_{i}(t-1))\nonumber$}
		\ENDIF
		\STATE {\textbf{Estimate consensus}:   $\hat x_{i,0}(t)=\tilde x_{i}(t), \mathcal{I}_{i,0}(t)=\mathcal{I}_{i}(t)$}\\		\vskip 2pt
			\FOR{$l=1,\dots,L$}
	\STATE{  Sensor $i$ obtains $\{\hat x_{j,l-1}(t),\mathcal{I}_{j,l-1}(t)\}$ from   sensor $j$,  \\		\vskip 2pt
		 $\hat x_{i,l}(t)=\hat x_{i,l-1}(t)-\alpha\sum_{j\in\mathcal{N}_{i}}(\hat x_{i,l-1}(t)-\hat x_{j,l-1}(t))$\\	
		\vskip 2pt
	 $\mathcal{I}_{i,l}(t)=\bigcup_{j\in\mathcal{N}_i}\mathcal{I}_{j,l-1}(t) \bigcup \mathcal{I}_{i,l-1}(t)$\\		\vskip 2pt
	}
\ENDFOR
		\STATE {Let $\hat x_{i}(t)=\hat x_{i,L}(t), \mathcal{I}_{i}(t)=\mathcal{I}_{i,L}(t).$ }
		\ENDFOR
	\end{algorithmic}
\end{algorithm}

{\begin{proposition}\label{lem_bound_normal}
	Consider Algorithm~\ref{alg:B} under the same setting as in Algorithm~\ref{alg:A}. 
			Under the same conditions as in  Theorem~\ref{thm_stability} and Assumption~\ref{ass_number}, the following results hold  for every sensor $i$ and every time $t\geq 1$:
		\begin{enumerate}
			\item 	The observation innovation of each attack-free sensor is upper bounded, i.e.,
			\begin{align}\label{eq_normal2}
			|y_{i}(t)-C_{i}A\hat x_{i}(t-1)|\leq \bar\varphi_i(t), i\in\mathcal{A}^c; 
			\end{align}
			\item  
			The sensors in  set $\mathcal{I}_{i}(t)$ are under attack for sure, i.e., $\mathcal{I}_{i}(t)\subseteq \mathcal{A}$; 
			\item   $\bar\varphi_i(t)$ is monotonically decreasing w.r.t. the number of the detected sensors (i.e., $d_i(t-1)$).
		\end{enumerate}
\end{proposition}}
\begin{proof}
	We  provide the proof idea as follows. To prove \eqref{eq_normal2}, we just need to show $\bar\rho_{t-1,i}\geq \norm{\tilde e(t-1)}$, where $\tilde e(t-1)=\|\frac{1}{N}\sum_{i=1}^N\hat x_{i}(t-1)-x(t-1)\|.$ This is done by referring to \eqref{eq_error_track_simple}--\eqref{eq_transition} and by noting that the  number of  attacked but undetected sensors   is upper bounded by $s-d_i(t-1)$. The conclusion 2)  follows from 1). The conclusion	3) is satisfied, since $\bar\rho_{t-1,i}$ in \eqref{sequence3} is monotonically decreasing w.r.t. $d_i(t-1)$.
\end{proof}

{
	
	In our framework, an attack signal $a_i(t)$ is stealthy if the compromised observation $y_{i}(t)$ violates   detection condition \eqref{detector2}, i.e., a stealthy attack signal $a_i(t)$ is in   set $\{a_i(t)\in\mathbb{R}|\hspace{2 pt} |\xi_i(t)+a_i(t)|\leq\bar\varphi_i(t)\}$, where   $\xi_i(t)$ is the attack-free observation innovation, i.e., $\xi_i(t)=C_iA(x(t-1)-\hat x(t-1))+C_iw(t-1)+v(t)$. Since this paper considers   bounded noise processes, a single large attack signal (larger than $\bar\varphi_i(t)$) will expose the attacked sensor for sure. In other words,	if there is a time $t$ such that $a_i(t) \in \{a_i(t)\in\mathbb{R}|\hspace{2 pt}|\xi_i(t)+a_i(t)|>\bar\varphi_i(t)\}$,    this attacked sensor $i$ will  be detected. Thus, this detection scheme  differs from the methods based on constructing statistical variables for  hypothesis tests  on the  innovation distributions (e.g., \cite{guo2017optimal}).  In our framework, the knowledge of the attacker on the designed detector especially on   threshold $\bar\varphi_i(t)$ will largely influence the detected sensor number.

	}

\subsection{Error Bounds and Convergence}

 Denote   $d(t)$ the maximal number of detected sensors at time $t$, i.e., $d(t)=\max_{i\in\mathcal{V}}\{d_i(t)\}$, with     $d(0)=0$.  Given any $T\geq 0,$ for $\forall t\geq T$, we construct the 
 following sequence $\{\bar \rho_{t}\in\mathbb{R}|\bar \rho_{t},t\geq T\}$
 \begin{align}\label{sequence2}
 \bar \rho_{t+1}=F(\bar \rho_{t})\bar \rho_{t}+\bar q_0,\qquad \bar\rho_T=\rho_T,
 \end{align}
 where $\bar q_0=q_0-\frac{d(T)\beta}{N}$, $\rho_T$ and $F(\cdot)$ are given in \eqref{sequence} and \eqref{notation}, respectively. Then, the following theorem builds the connection between $d(T)$ and an upper bound of the estimation error of Algorithm~\ref{alg:B}.

\begin{theorem}\label{thm_bounds_detected}
		Consider Algorithm~\ref{alg:B} under the same setting as in Algorithm~\ref{alg:A}. 
	Under the same conditions as in  Theorem~\ref{thm_stability} and Assumption~\ref{ass_number}, the estimation error     is asymptotically upper bounded, i.e., 
\begin{align*}
	&\limsup\limits_{t\rightarrow \infty}\norm{e_{i}(t)}\leq W(T), \forall T\geq 0,
\end{align*}
where 
\begin{align*}
W(T)&=\inf\limits_{t_0\in \Gamma} \rho_{ t_0}+\frac{\sqrt{N}\beta\gamma^L}{1-\norm{A}\gamma^L}-\frac{d(T)\beta}{N(1-F_{*})},
\end{align*}
and  $F_{*}=\inf\limits_{t_0\in \bar\Gamma}F( \bar\rho_{ t_0})\in [0,1)$, $\bar \Gamma=\{t\geq T|\bar\rho_{t}\leq \bar\rho_{t-1}\}$. 
\end{theorem}
\begin{proof}
See  Appendix \ref{pf_thm_bounds_detected}.
\end{proof}

{Algorithm~\ref{alg:B} ensures that   detected sensor number  $d(T)$ is non-decreasing as $T$ increases.
	As more sensors are detected (i.e., $d(T)$ is increasing), we get a tighter bound in  Theorem~\ref{thm_bounds_detected}. Thus, the bound in  Theorem~\ref{thm_bounds_detected} is  equal to or smaller than the bound in 3) of Theorem~\ref{thm_stability}. In practice, the system defender can use $d(T)$ at different $T$ to	generate a sequence of non-increasing bounds, which can be used to estimate the asymptotic error bound.  }
In the following theorem, we provide the conditions such that the state estimate of Algorithm~\ref{alg:B}  converges to the system state asymptotically.
\begin{theorem}\label{thm_observer}
(\textbf{Convergence})		Consider Algorithm~\ref{alg:B} under the same setting as in Algorithm~\ref{alg:A}. 	Under the same conditions as in  Theorem~\ref{thm_stability}, Assumption~\ref{ass_number}, and the following conditions 
	\begin{enumerate}
				\item the system is   noise-free, i.e., $w(t)\equiv 0$ and $v_i(t)\equiv 0$ for any $i\in\mathcal{V}$;
		\item {the attacker compromises $s$ sensors and they are detected in  finite time, i.e., there exists a finite time $\hat t_0$ and an $i$ such that $d_i(\hat t_0)=s$}; 
		\item   matrix $G$ is Schur stable,
	\end{enumerate}
 then the estimate in Algorithm~\ref{alg:B} will asymptotically converge to the state, i.e., 
		\begin{align*}
		\lim\limits_{t\rightarrow \infty}\norm{\hat x_i(t)-x(t)}=0, i\in\mathcal{V},
		\end{align*}
		where 
		\begin{align*} 
		\begin{split}
		G&=\begin{pmatrix}
		2\norm{A}\gamma^L&\norm{A}\gamma^L\sqrt{N-s}\\
			\tau_0&	\varpi
		\end{pmatrix}\\
\tau_0&=\max_{\mathcal{M}\subset\mathcal{V},|\mathcal{M}|=N-s}\norm{\frac{1}{N}(\textbf{1}_N^{\sf T}\otimes  I_n)\bar C^{\sf T}\bar K_{\mathcal{M}}\bar C(I_N\otimes A)}\\
\bar K_{\mathcal{M}}&=\diag\{\mathbb I_{1\in \mathcal{M}},\mathbb I_{2\in \mathcal{M}},\dots,\mathbb I_{N\in \mathcal{M}}\}\in\mathbb{R}^{N\times N},
		\end{split}
		\end{align*}
		where $\varpi$ is given in Theorem~\ref{thm_normal}, and $\bar C=\diag\{C_1,\dots,C_N\}$.
\end{theorem}
\begin{proof}
See  Appendix \ref{pf_thm_observer}.
\end{proof}

  {Condition 2) can be satisfied when the attacker compromises $s$ sensors  without using persistently  stealthy attack signals. In this case, the detector is able to identify all attacked sensors in finite time and remove their influence. Otherwise, the attacker can   have influence to the estimation  such that the estimation error is not tending to zero, 
  	but the estimation error is still bounded as we  state in Theorems~\ref{thm_stability}  and \ref{thm_bounds_detected}.
  	It is worth  noting that condition 2) can be  monitored by each sensor to know whether the   number of  detected sensors reaches $s$.}
  	Condition~3) can be fulfilled if the communication rate $L$ is relatively large, the number of the attack-free sensors (i.e., $N-s$) is relatively large, or the norm of the transition matrix (i.e., $\norm{A}$) is small.

\section{Simulation Results}\label{sec_simu}
In this section, we  provide   numerical simulations to show the effectiveness of the developed results.

{
	Consider  a second-order system monitored by a sensor network with 30 nodes, where   $A=\left(\begin{smallmatrix}
	1 & 0.1\\
	0 &1
	\end{smallmatrix}\right)$,   $C_i=(0,1)$ for sensor $i\in\{2,3,5,9,12,16,17,21,22,25,26,29\}$, and $C_i=(1,0)$ for the rest of the sensors. All observation noise $v_i(t) $  and 
	   process noise $w_j(t)$    follow the uniform distribution in $[0,0.01]$, where $(w_1(t),w_2(t))^{\sf T}=:w(t)$. The initial state is $x(0)=(25,25)^{\sf T}$. Each element of     $\hat x(0)$ follows the uniform distribution in $[0,25]$.  The bounds in Assumption~\ref{ass_bounds} are assumed to be $b_v=0.01, b_w=0.02,\eta_0=50.$ 
	We consider the time interval $t=0,1,\dots,500$,} and suppose that the attacker   inserts   signal  $a_i(t)=2(C_ix(t)+v_i(t))$  if   sensor $i$ is under attack.
 	We conduct a Monte Carlo experiment with $100$ runs. 
Define the average estimation error of   sensor $i$ and the maximum error over the whole network   by
\begin{align*}
\eta_i(t)&=\frac{1}{100}\sum_{j=1}^{100}\norm{e_i^j(t)},\\
\eta_{\max}(t)&=\frac{1}{100}\sum_{j=1}^{100}\max_{i\in\{1,\dots,30\}}\norm{e_i^j(t)},
\end{align*}
respectively, where $e_i^j(t)$ is the state estimation error of sensor $i$ at time $t$ in the $j$-th run.

\subsection{Secured Distributed Estimation}
In this subsection, we    verify the performance of Algorithm~\ref{alg:A} by considering the case that the attacked sensor set is time-varying. 
Assume that the distributed sensor networks under sensor attacks are switched in the  way of Fig.~\ref{fig:attack_networks2}  in each run, where a node in red means the node is under attack.

 By selecting $\beta=3$ and $L=3$ for Algorithm~\ref{alg:A},     estimation errors  $\eta_i(t),i=3,5,10,26$ are provided in   Fig.~\ref{fig:connections}~(a). From this figure, we see that there are three time instants, i.e., {$t=125,250,375$}, around which the error dynamics of the plotted sensors fluctuate. 
The reason why the error dynamics of one sensor increase lies in two aspects. First, after a certain time instant, this sensor is under attack. If so, its observations will be compromised and the estimation performance of this sensor will be degraded. For example, in   Fig.~\ref{fig:connections}~(a), the errors of sensor $5$ after time $t=375$, sensor $10$ after time $t=125$, and sensor $26$ after time $t=250$ all increase for this reason.    The second aspect is that   after a certain time instant, the neighbor of one senor is under attack. For this sensor, its error will also increase due to the consensus influence. For example, in   Fig.~\ref{fig:connections}~(a), the error of sensor $3$  increases after time $t=125$, since a neighbor of sensor $3$, i.e.,  sensor $10$, is under attack after time $t=125$.
{
	The reason that the estimation error of
	sensor 26 is  large  in the time interval 250--375  in    Fig.~\ref{fig:connections}~(a) is because of two reasons: 1) After $t=250$, sensor 26 is persistently under attack until $t=375$; 2) The   communication rate $L=1$ is small.
	Since the state estimate is affected by the compromised sensor observations, the estimation error will inevitably increase if  the information of neighbors  is not available in time. 	
	However, Algorithm~\ref{alg:A} provides a way to improve the estimation performance by increasing the communication rate $L$. 
	In   Fig.~\ref{fig:connections}~(b), the relationship between   estimation error $\eta(t)$ and   communication rate $L$ is studied.  The result shows that with the increase of communication rate $L$,   estimation error $\eta(t)$ decreases.}
The connection between   estimation  error $\eta(t)$ and   parameter  $\beta$ is studied in  Fig.~\ref{fig:connections}~(c) under  $L=4$. The figure shows that the estimation error with $\beta=5$ is the smallest within the errors by respectively setting $\beta=0.1,5,2000$. The result conforms to the discussion on the design of $\beta$ which should not be too small or too large in the algorithm setting.

To show the resilience of Algorithm~\ref{alg:A}, for the case that the attacked sensor set is time-invariant, under $L=4$ and $\beta=5$, we study the relationship between $|\mathcal{A}|$ and $\eta(t)$ in Fig.~\ref{fig:resilience} by randomly choosing a subset of the whole sensors (i.e., $\mathcal{A}$).  The result of this figure shows that with the increase of $|\mathcal{A}|$, the estimation error $\eta(t)$   becomes larger. Also, when  $|\mathcal{A}|$ is equal to or larger than half sensors, the estimation error is unstable.

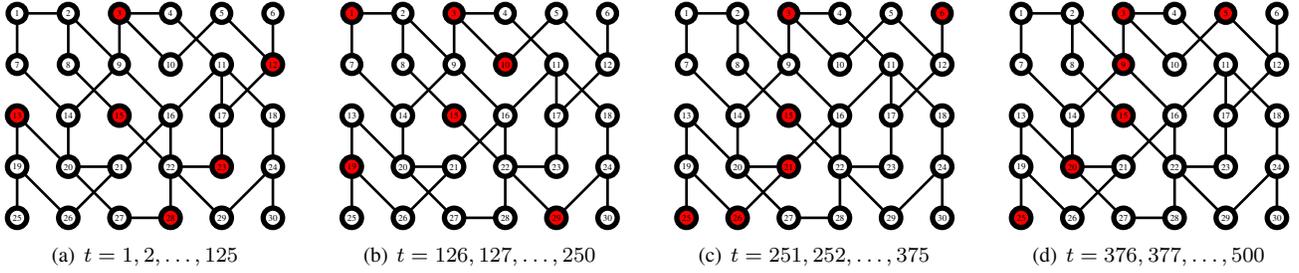
\begin{figure*}[t]
	\centering
	\subfigure[$t=1,2,\dots,125$]{
		\begin{tikzpicture}[scale=0.34, transform shape,line width=1pt]
		\node [draw,shape=circle,line width=2pt,minimum size=0.7cm] (1) at (0, 0) {1};
		\node[draw,shape=circle,line width=2pt,minimum size=0.7cm] (2) at +(2*1,0) {2};
		\node[draw,shape=circle,line width=2pt,fill=red,minimum size=0.7cm] (3) at +(2*2,0) {3};
		\node[draw,shape=circle,line width=2pt,minimum size=0.7cm] (4) at +(2*3,0) {4};
		\node[draw,shape=circle,line width=2pt,minimum size=0.7cm] (5) at +(2*4,0) {5};
		\node[draw,shape=circle,line width=2pt,minimum size=0.7cm] (6) at +(2*5,0) {6};
		\node[draw,shape=circle,line width=2pt,minimum size=0.7cm] (7) at (0, -2) {7};
		\node[draw,shape=circle,line width=2pt,minimum size=0.7cm] (8) at +(2,-2) {8};
		\node[draw,shape=circle,line width=2pt,minimum size=0.7cm] (9) at +(2*2,-2) {9};
		\node[draw,shape=circle,line width=2pt,minimum size=0.7cm] (10) at +(2*3,-2) {10};
		\node[draw,shape=circle,line width=2pt,minimum size=0.7cm] (11) at +(2*4,-2) {11};
		\node[draw,shape=circle,line width=2pt,fill=red,minimum size=0.7cm] (12) at +(2*5,-2) {12};
		\node[draw,shape=circle,line width=2pt,fill=red,minimum size=0.7cm] (13) at (0, -2*2) {13};
		\node[draw,shape=circle,line width=2pt,minimum size=0.7cm] (14) at +(2,-2*2) {14};
		\node[draw,shape=circle,line width=2pt,fill=red] (15) at +(2*2,-2*2) {15};
		\node[draw,shape=circle,line width=2pt,minimum size=0.7cm] (16) at +(2*3,-2*2) {16};
		\node[draw,shape=circle,line width=2pt,minimum size=0.7cm] (17) at +(2*4,-2*2) {17};
		\node[draw,shape=circle,line width=2pt,minimum size=0.7cm] (18) at +(2*5,-2*2) {18};
		\node[draw,shape=circle,line width=2pt,minimum size=0.7cm] (19) at (0, -2*3) {19};
		\node[draw,shape=circle,line width=2pt,minimum size=0.7cm] (20) at +(2,-2*3) {20};
		\node[draw,shape=circle,line width=2pt,minimum size=0.7cm] (21) at +(2*2,-2*3) {21};
		\node[draw,shape=circle,line width=2pt,minimum size=0.7cm] (22) at +(2*3,-2*3) {22};
		\node[draw,shape=circle,line width=2pt,fill=red] (23) at +(2*4,-2*3) {23};
		\node[draw,shape=circle,line width=2pt,minimum size=0.7cm] (24) at +(2*5,-2*3) {24};
		\node[draw,shape=circle,line width=2pt,minimum size=0.7cm] (25) at (0, -2*4) {25};
				\node[draw,shape=circle,line width=2pt,minimum size=0.7cm] (26) at +(2,-2*4) {26};
		\node[draw,shape=circle,line width=2pt,minimum size=0.7cm] (27) at +(2*2,-2*4) {27};
		\node[draw,shape=circle,line width=2pt,fill=red,minimum size=0.7cm] (28) at +(2*3,-2*4) {28};
		\node[draw,shape=circle,line width=2pt,minimum size=0.7cm] (29) at +(2*4,-2*4) {29};
		\node[draw,shape=circle,line width=2pt,minimum size=0.7cm] (30) at +(2*5,-2*4) {30};
		\foreach \from/\to in {1/2,1/7,2/8,2/9, 3/4,3/9,3/10,4/11, 5/10, 5/12,6/12,7/14,8/15,9/14,9/16,11/16,11/17,11/18,13/19,13/20,14/20,15/22,16/22,17/12,17/23,18/24,19/26,19/25,20/21,20/27,21/26,21/16,22/28,22/23,22/29,24/29,24/30,27/28}
		\draw [black] (\from) -- (\to);
		\end{tikzpicture}
	}
	\hskip 10pt
	\subfigure[$t=126,127,\dots,250$]{
		\begin{tikzpicture}[scale=0.34, transform shape,line width=1pt]
		\node [draw,shape=circle,line width=2pt,fill=red,minimum size=0.7cm] (1) at (0, 0) {1};
		\node[draw,shape=circle,line width=2pt,minimum size=0.7cm] (2) at +(2*1,0) {2};
		\node[draw,shape=circle,line width=2pt,fill=red,minimum size=0.7cm] (3) at +(2*2,0) {3};
		\node[draw,shape=circle,line width=2pt,minimum size=0.7cm] (4) at +(2*3,0) {4};
		\node[draw,shape=circle,line width=2pt,minimum size=0.7cm] (5) at +(2*4,0) {5};
		\node[draw,shape=circle,line width=2pt,minimum size=0.7cm] (6) at +(2*5,0) {6};
		\node[draw,shape=circle,line width=2pt,minimum size=0.7cm] (7) at (0, -2) {7};
		\node[draw,shape=circle,line width=2pt,minimum size=0.7cm] (8) at +(2,-2) {8};
		\node[draw,shape=circle,line width=2pt,minimum size=0.7cm] (9) at +(2*2,-2) {9};
		\node[draw,shape=circle,line width=2pt,fill=red,minimum size=0.7cm] (10) at +(2*3,-2) {10};
		\node[draw,shape=circle,line width=2pt,minimum size=0.7cm] (11) at +(2*4,-2) {11};
		\node[draw,shape=circle,line width=2pt,minimum size=0.7cm] (12) at +(2*5,-2) {12};
		\node[draw,shape=circle,line width=2pt,minimum size=0.7cm] (13) at (0, -2*2) {13};
		\node[draw,shape=circle,line width=2pt,minimum size=0.7cm] (14) at +(2,-2*2) {14};
		\node[draw,shape=circle,line width=2pt,fill=red] (15) at +(2*2,-2*2) {15};
		\node[draw,shape=circle,line width=2pt,minimum size=0.7cm] (16) at +(2*3,-2*2) {16};
		\node[draw,shape=circle,line width=2pt,minimum size=0.7cm] (17) at +(2*4,-2*2) {17};
		\node[draw,shape=circle,line width=2pt,minimum size=0.7cm] (18) at +(2*5,-2*2) {18};
		\node[draw,shape=circle,line width=2pt,fill=red,minimum size=0.7cm] (19) at (0, -2*3) {19};
		\node[draw,shape=circle,line width=2pt,minimum size=0.7cm] (20) at +(2,-2*3) {20};
		\node[draw,shape=circle,line width=2pt,minimum size=0.7cm] (21) at +(2*2,-2*3) {21};
		\node[draw,shape=circle,line width=2pt,minimum size=0.7cm] (22) at +(2*3,-2*3) {22};
		\node[draw,shape=circle,line width=2pt] (23) at +(2*4,-2*3) {23};
		\node[draw,shape=circle,line width=2pt,minimum size=0.7cm] (24) at +(2*5,-2*3) {24};
		\node[draw,shape=circle,line width=2pt,minimum size=0.7cm] (25) at (0, -2*4) {25};
				\node[draw,shape=circle,line width=2pt,minimum size=0.7cm] (26) at +(2,-2*4) {26};
		\node[draw,shape=circle,line width=2pt,minimum size=0.7cm] (27) at +(2*2,-2*4) {27};
		\node[draw,shape=circle,line width=2pt,minimum size=0.7cm] (28) at +(2*3,-2*4) {28};
		\node[draw,shape=circle,line width=2pt,fill=red,minimum size=0.7cm] (29) at +(2*4,-2*4) {29};
		\node[draw,shape=circle,line width=2pt,minimum size=0.7cm] (30) at +(2*5,-2*4) {30};
		\foreach \from/\to in {1/2,1/7,2/8,2/9, 3/4,3/9,3/10,4/11, 5/10, 5/12,6/12,7/14,8/15,9/14,9/16,11/16,11/17,11/18,13/19,13/20,14/20,15/22,16/22,17/12,17/23,18/24,19/26,19/25,20/21,20/27,21/26,21/16,22/28,22/23,22/29,24/29,24/30,27/28}
		\draw [black] (\from) -- (\to);
		\end{tikzpicture}	
	}
	\hskip 10pt
	\subfigure[$t=251,252,\dots,375$]{
		\begin{tikzpicture}[scale=0.34, transform shape,line width=1pt]
		\node [draw,shape=circle,line width=2pt,minimum size=0.7cm] (1) at (0, 0) {1};
		\node[draw,shape=circle,line width=2pt,minimum size=0.7cm] (2) at +(2*1,0) {2};
		\node[draw,shape=circle,line width=2pt,fill=red,minimum size=0.7cm] (3) at +(2*2,0) {3};
		\node[draw,shape=circle,line width=2pt,minimum size=0.7cm] (4) at +(2*3,0) {4};
		\node[draw,shape=circle,line width=2pt,minimum size=0.7cm] (5) at +(2*4,0) {5};
		\node[draw,shape=circle,line width=2pt,fill=red,minimum size=0.7cm] (6) at +(2*5,0) {6};
		\node[draw,shape=circle,line width=2pt,minimum size=0.7cm] (7) at (0, -2) {7};
		\node[draw,shape=circle,line width=2pt,minimum size=0.7cm] (8) at +(2,-2) {8};
		\node[draw,shape=circle,line width=2pt,minimum size=0.7cm] (9) at +(2*2,-2) {9};
		\node[draw,shape=circle,line width=2pt,minimum size=0.7cm] (10) at +(2*3,-2) {10};
		\node[draw,shape=circle,line width=2pt,minimum size=0.7cm] (11) at +(2*4,-2) {11};
		\node[draw,shape=circle,line width=2pt,minimum size=0.7cm] (12) at +(2*5,-2) {12};
		\node[draw,shape=circle,line width=2pt,minimum size=0.7cm] (13) at (0, -2*2) {13};
		\node[draw,shape=circle,line width=2pt,minimum size=0.7cm] (14) at +(2,-2*2) {14};
		\node[draw,shape=circle,line width=2pt,fill=red] (15) at +(2*2,-2*2) {15};
		\node[draw,shape=circle,line width=2pt,minimum size=0.7cm] (16) at +(2*3,-2*2) {16};
		\node[draw,shape=circle,line width=2pt,minimum size=0.7cm] (17) at +(2*4,-2*2) {17};
		\node[draw,shape=circle,line width=2pt,minimum size=0.7cm] (18) at +(2*5,-2*2) {18};
		\node[draw,shape=circle,line width=2pt,minimum size=0.7cm] (19) at (0, -2*3) {19};
		\node[draw,shape=circle,line width=2pt,minimum size=0.7cm] (20) at +(2,-2*3) {20};
		\node[draw,shape=circle,line width=2pt,fill=red,minimum size=0.7cm] (21) at +(2*2,-2*3) {21};
		\node[draw,shape=circle,line width=2pt,minimum size=0.7cm] (22) at +(2*3,-2*3) {22};
		\node[draw,shape=circle,line width=2pt] (23) at +(2*4,-2*3) {23};
		\node[draw,shape=circle,line width=2pt,minimum size=0.7cm] (24) at +(2*5,-2*3) {24};
		\node[draw,shape=circle,line width=2pt,fill=red,minimum size=0.7cm] (25) at (0, -2*4) {25};
				\node[draw,shape=circle,line width=2pt,fill=red,minimum size=0.7cm] (26) at +(2,-2*4) {26};
		\node[draw,shape=circle,line width=2pt,minimum size=0.7cm] (27) at +(2*2,-2*4) {27};
		\node[draw,shape=circle,line width=2pt,minimum size=0.7cm] (28) at +(2*3,-2*4) {28};
		\node[draw,shape=circle,line width=2pt,minimum size=0.7cm] (29) at +(2*4,-2*4) {29};
		\node[draw,shape=circle,line width=2pt,minimum size=0.7cm] (30) at +(2*5,-2*4) {30};
		\foreach \from/\to in {1/2,1/7,2/8,2/9, 3/4,3/9,3/10,4/11, 5/10, 5/12,6/12,7/14,8/15,9/14,9/16,11/16,11/17,11/18,13/19,13/20,14/20,15/22,16/22,17/12,17/23,18/24,19/26,19/25,20/21,20/27,21/26,21/16,22/28,22/23,22/29,24/29,24/30,27/28}
		\draw [black] (\from) -- (\to);
		\end{tikzpicture}	
	}
	\hskip 10pt
	\subfigure[$t=376,377,\dots,500$]{
		\begin{tikzpicture}[scale=0.34, transform shape,line width=1pt]
		\node [draw,shape=circle,line width=2pt,minimum size=0.7cm] (1) at (0, 0) {1};
		\node[draw,shape=circle,line width=2pt,minimum size=0.7cm] (2) at +(2*1,0) {2};
		\node[draw,shape=circle,line width=2pt,fill=red,minimum size=0.7cm] (3) at +(2*2,0) {3};
		\node[draw,shape=circle,line width=2pt,minimum size=0.7cm] (4) at +(2*3,0) {4};
		\node[draw,shape=circle,line width=2pt,fill=red,minimum size=0.7cm] (5) at +(2*4,0) {5};
		\node[draw,shape=circle,line width=2pt,minimum size=0.7cm] (6) at +(2*5,0) {6};
		\node[draw,shape=circle,line width=2pt,minimum size=0.7cm] (7) at (0, -2) {7};
		\node[draw,shape=circle,line width=2pt,minimum size=0.7cm] (8) at +(2,-2) {8};
		\node[draw,shape=circle,line width=2pt,fill=red,minimum size=0.7cm] (9) at +(2*2,-2) {9};
		\node[draw,shape=circle,line width=2pt,minimum size=0.7cm] (10) at +(2*3,-2) {10};
		\node[draw,shape=circle,line width=2pt,minimum size=0.7cm] (11) at +(2*4,-2) {11};
		\node[draw,shape=circle,line width=2pt,minimum size=0.7cm] (12) at +(2*5,-2) {12};
		\node[draw,shape=circle,line width=2pt,minimum size=0.7cm] (13) at (0, -2*2) {13};
		\node[draw,shape=circle,line width=2pt,minimum size=0.7cm] (14) at +(2,-2*2) {14};
		\node[draw,shape=circle,line width=2pt,fill=red] (15) at +(2*2,-2*2) {15};
		\node[draw,shape=circle,line width=2pt,minimum size=0.7cm] (16) at +(2*3,-2*2) {16};
		\node[draw,shape=circle,line width=2pt,minimum size=0.7cm] (17) at +(2*4,-2*2) {17};
		\node[draw,shape=circle,line width=2pt,minimum size=0.7cm] (18) at +(2*5,-2*2) {18};
		\node[draw,shape=circle,line width=2pt,minimum size=0.7cm] (19) at (0, -2*3) {19};
		\node[draw,shape=circle,line width=2pt,fill=red,minimum size=0.7cm] (20) at +(2,-2*3) {20};
		\node[draw,shape=circle,line width=2pt,minimum size=0.7cm] (21) at +(2*2,-2*3) {21};
		\node[draw,shape=circle,line width=2pt,minimum size=0.7cm] (22) at +(2*3,-2*3) {22};
		\node[draw,shape=circle,line width=2pt] (23) at +(2*4,-2*3) {23};
		\node[draw,shape=circle,line width=2pt,minimum size=0.7cm] (24) at +(2*5,-2*3) {24};
		\node[draw,shape=circle,line width=2pt,fill=red,minimum size=0.7cm] (25) at (0, -2*4) {25};
						\node[draw,shape=circle,line width=2pt,minimum size=0.7cm] (26) at +(2,-2*4) {26};
		\node[draw,shape=circle,line width=2pt,minimum size=0.7cm] (27) at +(2*2,-2*4) {27};
		\node[draw,shape=circle,line width=2pt,minimum size=0.7cm] (28) at +(2*3,-2*4) {28};
		\node[draw,shape=circle,line width=2pt,minimum size=0.7cm] (29) at +(2*4,-2*4) {29};
		\node[draw,shape=circle,line width=2pt,minimum size=0.7cm] (30) at +(2*5,-2*4) {30};
		\foreach \from/\to in {1/2,1/7,2/8,2/9, 3/4,3/9,3/10,4/11, 5/10, 5/12,6/12,7/14,8/15,9/14,9/16,11/16,11/17,11/18,13/19,13/20,14/20,15/22,16/22,17/12,17/23,18/24,19/26,19/25,20/21,20/27,21/26,21/16,22/28,22/23,22/29,24/29,24/30,27/28}
		\draw [black] (\from) -- (\to);
		\end{tikzpicture}	
	}
	\caption{Sensor network with different sets of attacked sensors over four time intervals.}
	\label{fig:attack_networks2}
\end{figure*}

\begin{figure*}[t]
	\centering
	\subfigure[The estimation errors of some sensors.]{
		\includegraphics[scale=0.39]{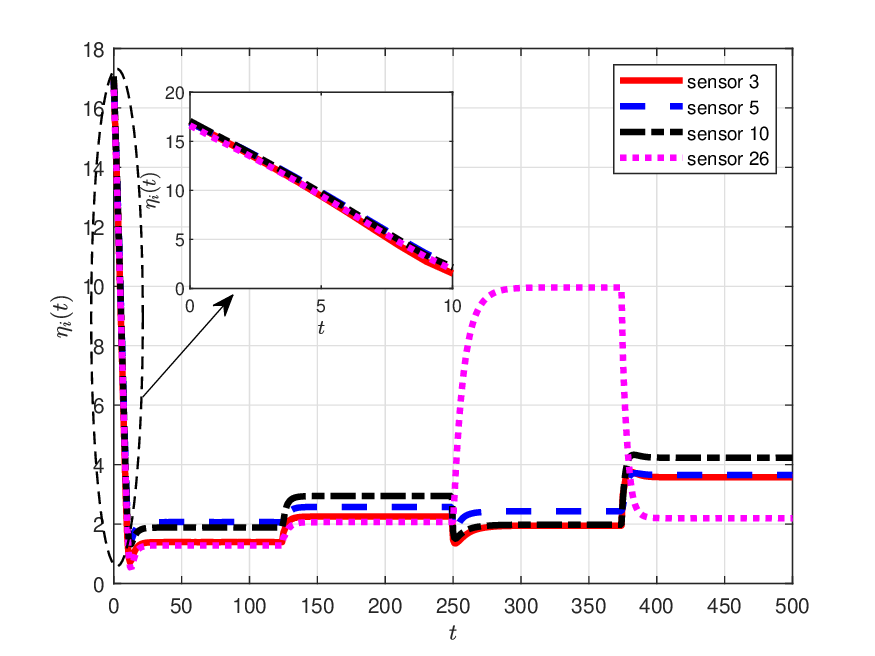}}
	\subfigure[\label{asyn}The relationship between $L$ and $\eta_{\max}(t)$.]{
		\includegraphics[scale=0.39]{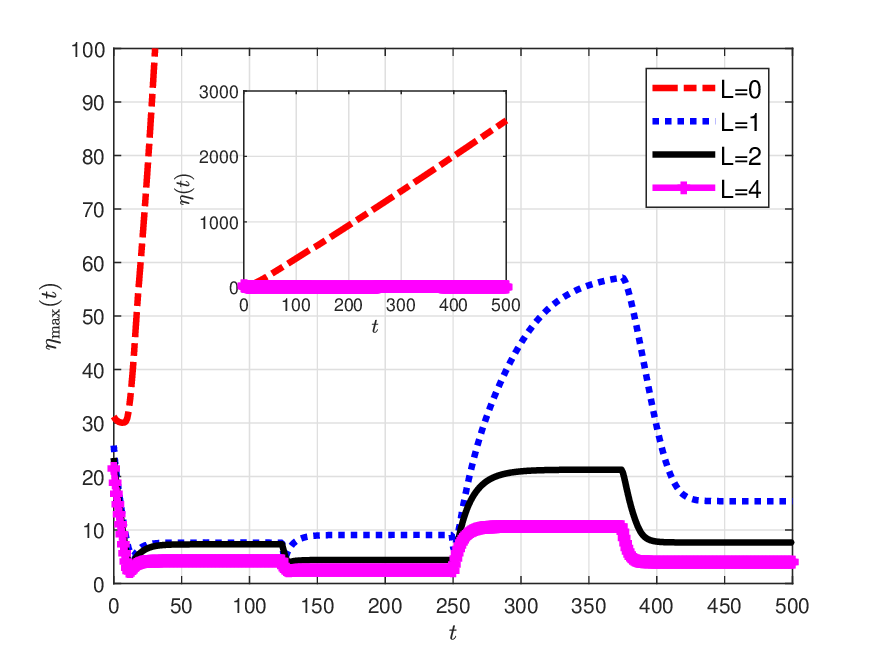}}
	\subfigure[\label{linkf}The relationship between $\beta$ and $\eta_{\max}(t)$.]{
		\includegraphics[scale=0.39]{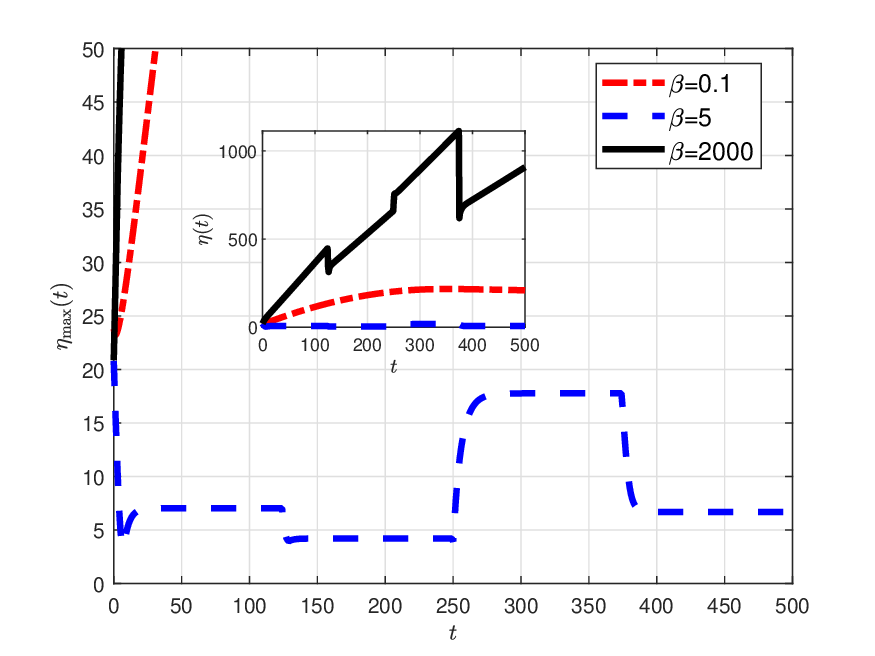}}
	\caption{Estimation performance of Algorithm~\ref{alg:A} for the sensor networks in Fig.~\ref{fig:attack_networks2}.}\label{fig:connections}
\end{figure*}
\begin{figure*}[t]
	\centering
	\subfigure[The estimation errors of some sensors.]{
		\includegraphics[scale=0.39]{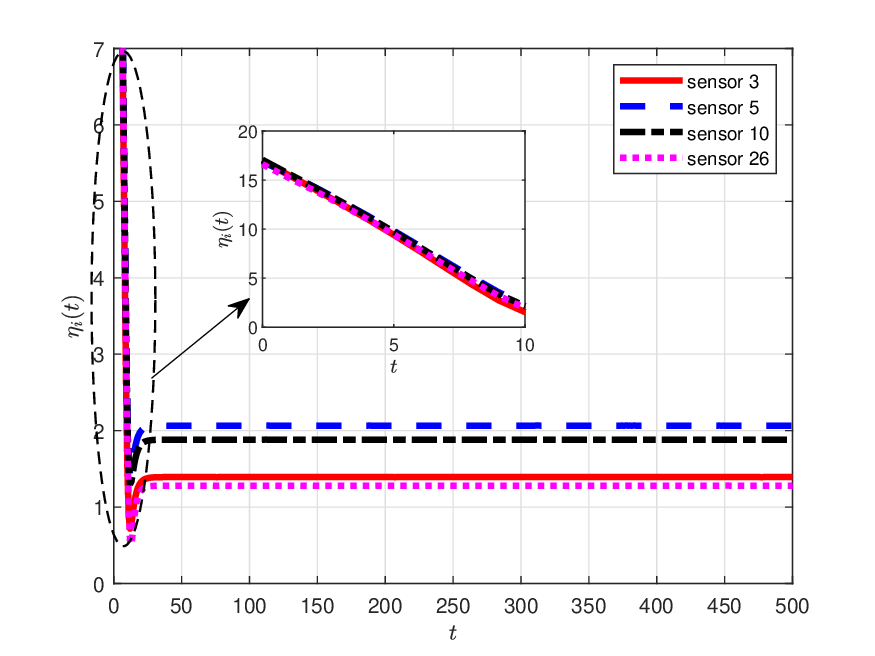}}
	\subfigure[\label{asyn2}The relationship between $L$ and $\eta_{\max}(t)$.]{
		\includegraphics[scale=0.39]{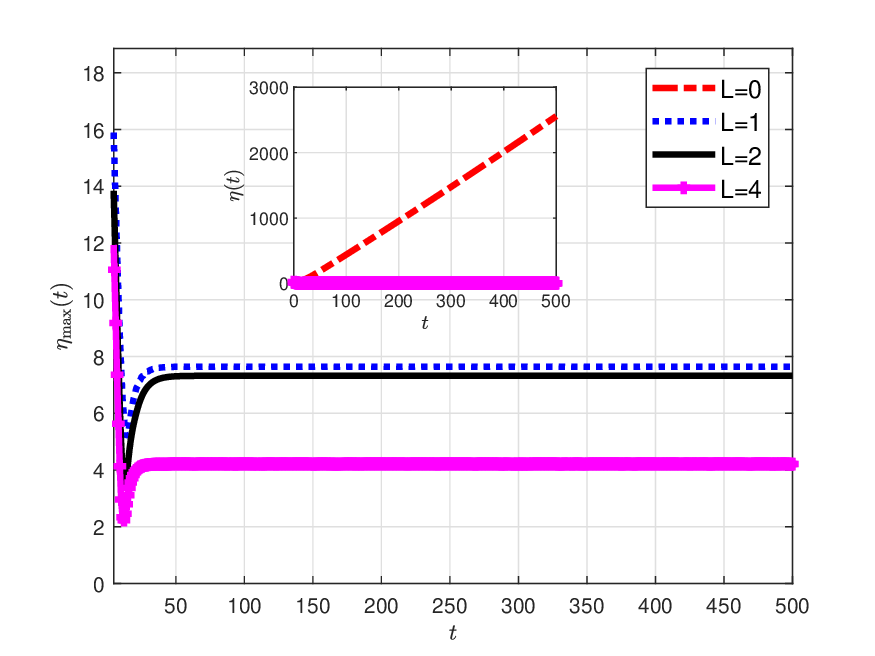}}
	\subfigure[\label{inertia2}The relationship between $\beta$ and $\eta_{\max}(t)$.]{
		\includegraphics[scale=0.39]{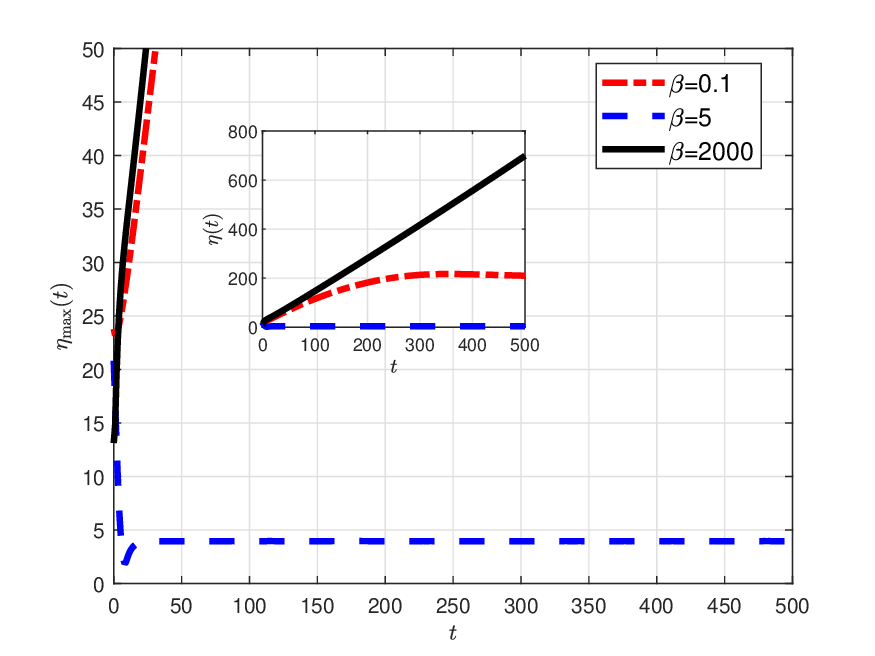}}
	\caption{Estimation performance of Algorithm~\ref{alg:A} for the sensor network in Fig.~\ref{fig:attack_networks}.}\label{fig:connections2}
\end{figure*}
\begin{figure}[t]
	\centering
	\includegraphics[scale=0.5]{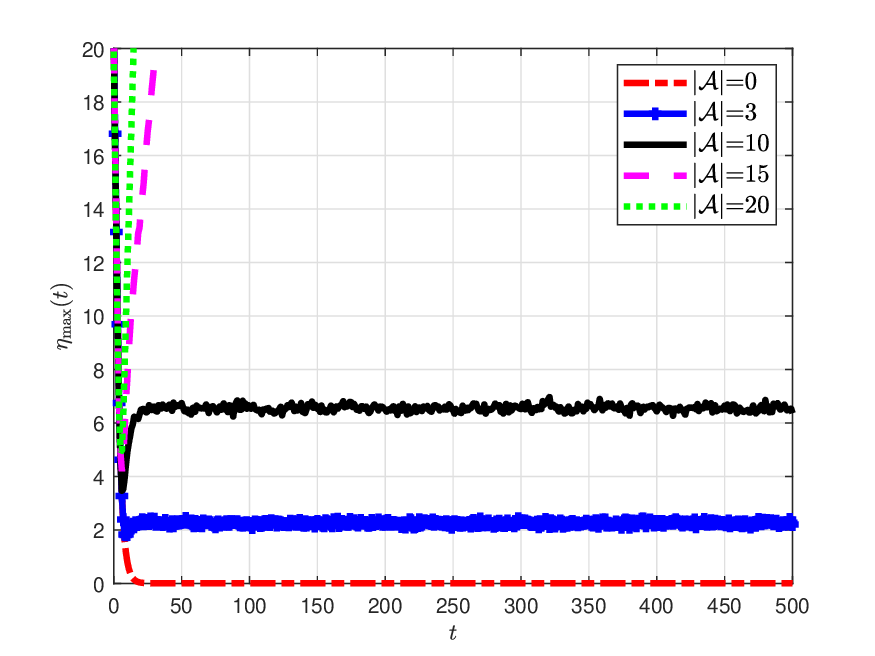}
	\caption {Resilience of Algorithm~\ref{alg:A}.}
	\label{fig:resilience}
\end{figure}
\begin{figure}[t]
	\centering
	\includegraphics[scale=0.5]{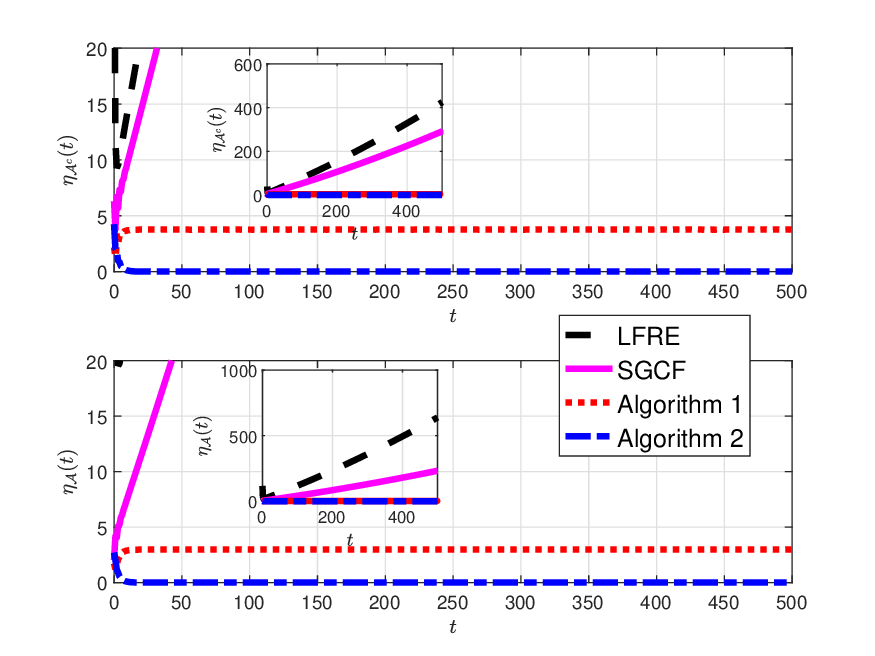}
	\caption {The estimation error comparison of algorithms.}
	\label{fig:detec}
\end{figure}
\subsection{Secured Distributed Estimation Under Detection}
In this subsection, we consider the introduction case in Fig.~\ref{fig:attack_networks} where the attacked sensor set is time-invariant, with $\mathcal{A}(t)=\mathcal{A}=\{3,12,13,15,23,28\},$ under which the estimation performance of  Algorithms \ref{alg:A} and \ref{alg:B} is studied. 

 By selecting $\beta=3$ and $L=3$ for Algorithm~\ref{alg:A},   estimation errors  $\eta_i(t),i=3,5,10,26$ are provided   in   Fig.~\ref{fig:connections2}~(a). Compared to the result in   Fig.~\ref{fig:connections}~(a), the estimation errors of these sensors  in  Fig.~\ref{fig:connections2}~(a) do no fluctuate. Since the attacked sensor set is time-invariant,  the   estimation errors become steady after   transience as shown in the figure. 
 The relationships between   parameters $L$, $\beta$ and    estimation error $\eta(t)$ are shown in   Fig.~\ref{fig:connections2}~(b) and (c). Note that in   Fig.~\ref{fig:connections2}~(b), when $L=0$,   estimation error $\eta(t)$ is divergent, since  for the attacked sensors, the local observability is violated. Fig.~\ref{fig:connections2}~(b) also shows that with the increase of   communication rate $L$, the estimation error of Algorithm~\ref{alg:A} is decreasing. In   Fig.~\ref{fig:connections2}~(c),   parameter $\beta=5$ can lead to the stable estimation error, while in  Fig.~\ref{fig:connections}~(c), the estimation error for   parameter $\beta=5$ fluctuates due to the switching of the attacked sensor set.

{In order to compare with existing algorithms under the situation in Fig.~\ref{fig:attack_networks}}, we define the estimation errors for the attacked sensor set and attack-free sensor set, respectively:  $\eta_{\mathcal{A}}(t)=\frac{1}{100}\sum_{j=1}^{100}\max_{i\in \mathcal{A}}\norm{e_i^j(t)}$, and $\eta_{\mathcal{A}^c}(t)=\frac{1}{100}\sum_{j=1}^{100}\max_{i\in \mathcal{A}^c}\norm{e_i^j(t)}.$
In Fig.~\ref{fig:detec}, the estimation performance of Algorithms 1 and 2 with $L=5$, the local-filtering based resilient estimation (LFRE) \cite{mitra2019byzantine}  and the scalar-gain consensus filter\footnote{The filter has the same form as Algorithm~\ref{alg:A} but   $k_i(t)=1$ for $t\geq 1.$ It follows the idea   in \cite{Khan2014Collaborative}, which did not consider the attack scenario.} (SGCF) is compared. The figure shows that Algorithms \ref{alg:A} and \ref{alg:B}   provide estimates with stable estimation errors for both attacked sensors and attack-free sensors, but the estimation errors of LFRE and SGCF are divergent. 
Compared to Algorithm~\ref{alg:A}, Algorithm~\ref{alg:B}  provides   smaller estimation errors by successfully detecting all the attacked sensors.  

\section{Conclusion and Future Work}\label{sec_conclusion}
This paper studied the   distributed filtering problem for linear time-invariant systems with bounded noise { under false-data injection attacks in sensors networks}, where a malicious attacker can compromise a time-varying and unknown subset of  sensors and manipulate their observations arbitrarily.
First, we proposed a   distributed saturation-based filter.
Then, we provided a sufficient condition to guarantee boundedness of the estimation error.  
By confining the attacked sensor set to be time-invariant, we  then  modified the filter by adding an    attack detection scheme.  Moreover, for the noise-free case, we proved that the state estimate of each sensor asymptotically converges to the system state under certain conditions. 

{There are some future directions. Since this paper employs a two-time-scale scheme in filter design, it is interesting to develop algorithms which use one communication at each time. Other directions include considering more general systems (e.g., nonlinear systems) and   more complex sensor networks (e.g., random or  communication-delayed networks). }

{\section*{ACKNOWLEDGMENT}
	The authors are grateful to  the anonymous reviewers for their insightful comments and suggestions.}

\appendix
\subsection{Proof of Lemma \ref{lem_stability}}\label{pf_lem_stability}
First, we prove 2). Due to $t_0\in\Gamma$, we have 
\begin{align}\label{lem_pf_1}
F(x_{t_0-1})x_{t_0-1}+q_0= x_{t_0}\leq x_{t_0-1}.
\end{align}
For $t=t_0+1$, by \eqref{lem_pf_1} and the condition that $F(\cdot)\in[0,1]$ is a monotonically non-decreasing function, we have
$x_{t_0+1}=F(x_{t_0})x_{t_0}+q_0 
\leq F(x_{t_0-1})x_{t_0-1}+q_0= x_{t_0}.$
By recursively applying the above procedure, we have $\sup_{t\geq t_0}x_t\leq x_{t_0}.$

Next, we prove 1), which trivially holds for $t=t_0$. Consider the case of $t>t_0$ in the following. 
If $q_0\neq 0$, it follows from \eqref{lem_pf_1} that $F(x_{t_0-1})\in [0,1)$. Due to $F(x_{t_0})\leq F(x_{t_0-1})$, we have $F(x_{t_0})\in [0,1)$. By 2) and the condition that $F(\cdot)\in[0,1]$ is a monotonically non-decreasing function, we have $F(x_{t-1})\leq F(x_{t_0})$. Then $x_{t}\leq F(x_{t_0})x_{t-1}+q_0$ with $F(x_{t_0})\in [0,1)$.
Thus, 1) is satisfied by recursively applying the inequality for $t-t_0$ times.

Finally, we prove 3). 
By 1) and the definition of $\limsup$, we have $\limsup\limits_{t\rightarrow \infty}x_t=\inf\limits_{t\in\mathbb{Z}^+}\sup\limits_{m\geq  t}x_m
\leq\inf\limits_{t\in\Gamma}\sup\limits_{m\geq  t}x_m
\leq 	\inf\limits_{t\in \Gamma} x_{ t}.$

\subsection{Proof of Theorem~\ref{thm_stability}}\label{app_A}
Let   $e_i(t)=\tilde e(t)+\bar e_i(t)$, 
where $\bar e_i(t):=\hat x_{i}(t)-\hat x_{avg}$, and $\tilde e(t)=\hat x_{avg}-x(t),$ and $\hat x_{avg}(t):=\frac{1}{N}\sum_{i=1}^N\hat x_{i}(t)$.
 Besides, we denote
\begin{align}\label{eq_denotations}
X(t)&=\textbf{1}_N\otimes  x(t)\in\mathbb{R}^{Nn},\nonumber\\
\bar E(t)&=\left(\bar e_1^{\sf T}(t),\dots,\bar e_N^{\sf T}(t)\right)^{\sf T}\in\mathbb{R}^{Nn},\nonumber\\
Y(t)&=\left(y_1^{\sf T}(t),\dots,y_N^{\sf T}(t)\right)^{\sf T}\in\mathbb{R}^{N},\nonumber\\
V(t)&=\left(v_1^{\sf T}(t),\dots,v_N^{\sf T}(t)\right)^{\sf T}\in\mathbb{R}^{N},\nonumber\\
\hat X(t)&=\left(\hat x_1^{\sf T}(t),\dots,\hat  x_N^{\sf T}(t)\right)^{\sf T}\in\mathbb{R}^{Nn},\\
\bar C&=\diag\{C_1,\dots,C_N\}\in\mathbb{R}^{N\times Nn},\nonumber\\
\bar K(t)&=\diag\{k_1(t),\dots,k_N(t)\}\in\mathbb{R}^{N\times N},\nonumber\\
P_{Nn}&=\frac{1}{N}(\textbf{1}_N\otimes  I_n)(\textbf{1}_N\otimes  I_n)^{\sf T}\in\mathbb{R}^{Nn\times Nn}.\nonumber
\end{align}

The idea for the proof is that we first  show  $\norm{\bar e_i(t)}$ is upper bounded by $p(t)$, and then we prove   $\norm{\tilde e(t)}$ is upper bounded by the quantities in 1)--3) of the theorem. The following lemma with   a similar proof as in \cite{He2020Secured}  ensures $\norm{\bar e_i(t)}\leq p(t)$. 
\begin{lemma}\label{prop_consensus}
	Consider Algorithm~\ref{alg:A}, and let Assumptions~\ref{ass_bounds}--\ref{ass_graph} hold.
	If $ \alpha=\frac{2}{\lambda_2(\mathcal{L})+\lambda_{\max}(\mathcal{L})}$, and $L>\frac{\ln \norm{A}}{\ln \gamma^{-1}}$,
	then for $t\geq 0,	$
	\begin{align}\label{eq_error_norm}
	\norm{\bar E(t)}
	\leq p(t),
	\end{align}
where $p(t)$ and $\bar E(t)$   are defined in  \eqref{eq_p_R} and \eqref{eq_denotations}, respectively.
\end{lemma}

\begin{figure*}[ht]
{
		\begin{align}\label{eq_notations}
		\begin{split}
		&	\line(1,0){470}\\
		&\hat x_{avg}(t)=A\hat x_{avg}(t-1)
		+ \frac{1}{N}(\textbf{1}_N^{\sf T}\otimes  I_n)P_{Nn}\left(I_{Nn}-\alpha(\mathcal{L}\otimes I_n)\right)^L
		\bar C^{\sf T}\bar K(t)h(t),\\
		&\bar K_{\mathcal{J}(t)}=\diag\{k_1(t)\mathbb I_{1\in \mathcal{J}(t)},\dots,k_N(t)\mathbb I_{N\in \mathcal{J}(t)}\},\quad 
		\bar K_{\mathcal{J}^c(t)}=\diag\{k_1(t)\mathbb I_{1\in \mathcal{J}^c(t)},\dots,k_N(t)\mathbb I_{N\in \mathcal{J}^c(t)}\},\\
		&\bar m_{t}=w(t-1)+\frac{1}{N}(\textbf{1}_N^{\sf T}\otimes  I_n)\bar C^{\sf T}\bar K_{\mathcal{J}(t)}(t)\left(\bar C((I_N\otimes A)\bar E(t-1)-I_{N}\otimes w(t-1))-V(t)\right)\\
		&h(t)=Y(t)-\bar C(I_N\otimes A)\hat X(t-1), \quad  M_{t}=(I_n-\frac{1}{N}\sum_{i\in\mathcal{J}(t)}k_{i}(t)C_{i}^{\sf T}C_{i})A,
		\quad \tilde m_{t}=\frac{1}{N}(\textbf{1}_N^{\sf T}\otimes  I_n) \bar C^{\sf T}\bar K_{\mathcal{J}^c(t)}(t)h(t)\\
		&\line(1,0){470}
		\end{split}
		\end{align}
	}
\end{figure*}

\textbf{Proof of Theorem~\ref{thm_stability}:}
{ It follows from \eqref{sequence} and \eqref{condition_thm}  that  $\rho_1\leq \rho_0=\eta_0$, which means that 
	$1\in\Gamma=\{t\geq 1|\rho_{t}\leq \rho_{t-1}\}$.} In the following, we prove 1)--3). 

{Under Assumption~\ref{ass_attacker}, there are at least $N-s$ attack-free sensors at each time. Suppose $\mathcal{J}(t)$ is the set of these $N-s$ sensors, i.e., $\mathcal{J}(t)\subseteq \mathcal{A}^c(t)$ with $|\mathcal{J}(t)|=N-s$. Denote $\mathcal{J}^c(t)=\mathcal{V}\setminus\mathcal{J}(t)$, which satisfies $|\mathcal{J}^c(t)|=s$ due to $|\mathcal{V}|=N.$  }
By Algorithm~\ref{alg:A} and the notations in \eqref{eq_denotations}, we have the dynamics of $\hat x_{avg}(t)$ in \eqref{eq_notations}.
Due to $(\textbf{1}_N^{\sf T}\otimes  I_n)P_{Nn}=(\textbf{1}_N^{\sf T}\otimes  I_n)\left(I_{Nn}-\alpha(\mathcal{L}\otimes I_n)\right)^L=(\textbf{1}_N^{\sf T}\otimes  I_n),$
we have
\begin{align}\label{eq_error_track_simple}
\tilde e(t)
=&M_{t}\tilde e(t-1)-\bar m_{t}+\tilde m_{t},
\end{align}
where $M_{t}$, $\tilde m_{t}$, and $\bar m_{t}$ are given in \eqref{eq_notations}.
Note that $\tilde m_{t}$ can be rewritten in the following way
\begin{align*}
\tilde m_{t}=\frac{1}{N}\sum_{i\in\mathcal{J}^c(t)}C_i^{\sf T}k_i(t)(y_{i}(t)-C_{i}A\hat x_{i}(t-1)).
\end{align*}
{Due to $k_{i}(t)=\min\{1,\frac{\beta}{|y_{i}(t)-C_{i}A\hat x_{i}(t-1)|}\}$, it holds that $|k_i(t)(y_{i}(t)-C_{i}A\hat x_{i}(t-1))|\leq \beta$. Since we assume $\|C_i\|=1$ after the system model, it holds that $\|\tilde m_{t}\|\leq \frac{1}{N}\sum_{i\in\mathcal{J}^c(t)}\|C_i^{\sf T}\|\beta\leq |\mathcal{J}^c(t)|\frac{\beta}{N}.$ Due to $|\mathcal{J}^c(t)|= s$, 
	we have 
	\begin{align}\label{eq_attack_error}
	\norm{\tilde m_{t}}\leq&  \frac{s}{N}\beta.
	\end{align}}
Regarding $\bar m_{t}$, by Assumption~$\ref{ass_attacker}$ and $k_{i}(t)\leq 1$, we have
\begin{align}\label{eq_condition}
\norm{\bar m_{t}}&\leq \norm{w(t-1)}+\frac{|\mathcal{J}(t)|}{N}(\norm{A}\norm{\bar E(t-1)}\nonumber\\
&\quad+\norm{w(t-1)}+b_v)\nonumber\\
&\leq \frac{N-s}{N}(b_{w}+b_{v}+\norm{A}p_0)+b_{w},
\end{align}
where the second inequality is obtained by Lemma \ref{prop_consensus} and $\sup_{t\geq 0}p(t)\leq p_0$, where $p_0$ is defined in \eqref{notation}.
Based on \eqref{eq_error_track_simple}--\eqref{eq_condition}, we construct the sequence $ \{\rho_{t}\}$ in \eqref{sequence}. 
In the following, we prove  that $\norm{\tilde e(t)}\leq\rho_{t}$.

At the initial time, i.e., $t=0,$ by Assumption~\ref{ass_bounds}, we have $\norm{\tilde e(0)}=\norm{ \hat x_{avg}(0)-x(0)}\leq \frac{1}{N}\sum_{i=1}^{N}\norm{\hat x_{i}(0)-x(0)}\leq \eta_0.$ Due to $\rho_0=\eta_0$, $\norm{\tilde e(t)}\leq\rho_{t}$ for $t=0.$ Suppose at time $t-1$, $\norm{\tilde e(t-1)}\leq \rho_{t-1}.$  
At time $t$, for $i\in\mathcal{J}(t)$, we consider  
\begin{align}\label{eq_inno}
&|y_{i}(t)-C_iA\hat x_{i}(t-1)|\nonumber\\
\leq&\norm{A}\norm{e_{i}(t-1)}+b_{w}+b_{v}\nonumber\\
\leq&\norm{A}(\norm{\bar e_{i}(t-1)}+\norm{\tilde e(t-1)})+b_{w}+b_{v}\nonumber\\
\leq&\norm{A}(p_0+\rho_{t-1})+b_{w}+b_{v},
\end{align}
where the last inequality of \eqref{eq_inno} is obtained by noting that
$\sup_{t\geq 0}\norm{\bar e_{i}(t)}\leq \sup_{t\geq 0}p(t)\leq p_0$ and $p_0$ is defined in \eqref{notation}.
Recall the form of $k_{i}(t)$, 
by (\ref{eq_inno}), for $i\in\mathcal{J}(t)$, we have $k_{i}(t)\geq k^*(\rho_{t-1}):=\min\bigg\{1,\frac{\beta}{\norm{A}(p_0+\rho_{t-1})+b_{w}+b_{v}}\bigg\}>0.$
 Then
\begin{align}\label{eq_transition}
\norm{M_{t}}\leq &\norm{A}\norm{\left(I_n-\frac{k^*(\rho_{t-1})}{N}\sum_{i\in\mathcal{J}(t)}C_{i}^{\sf T}C_{i}\right)}\nonumber\\
\leq&\norm{A}\left(1-\frac{k^*(\rho_{t-1})}{N}\lambda_0\right).
\end{align}
Taking norm   on both sides of \eqref{eq_error_track_simple} and considering \eqref{sequence}, \eqref{eq_attack_error},  \eqref{eq_condition},  and \eqref{eq_transition}, we have $\norm{\tilde e(t)}\leq \rho_{t}.$ 

Since the defined $F(\rho_t)$ in \eqref{notation} is monotonically non-decreasing function,   conclusions 1)--3) of this theorem are obtained by  applying the results in Lemma \ref{lem_stability}, $\norm{e_i(t)}\leq \norm{\tilde e(t)}+\norm{\bar e_i(t)}$, and \eqref{eq_error_norm}.

\subsection{Proof of Theorem~\ref{thm_iff}}\label{app_thm_iff}
{
	1) Sufficiency:
	
	Case 1: For the case  $s>0$, 	
	we consider $\norm{A}\in [1,1+\epsilon)$ with  $\epsilon=\frac{\lambda_0-s}{4(N-\lambda_0)}$ which is positive due to $\lambda_0>s$. 
	If $L$ is sufficiently large, $p_0>0$ in \eqref{notation} will be sufficiently small. Thus, given   noise bounds $b_{w}$ and $b_{v}$, considering $N\geq s+\lambda_0$,
	it is feasible to choose sufficiently large $\beta,\eta_0$ and $L>\frac{\ln \norm{A}}{\ln \gamma^{-1}}$, such that 
	\begin{align}\label{pf_beta}
	\begin{split}
	\beta&\geq (1+\epsilon)(p_0+\eta_0)+b_{w}+b_{v},\\
	\beta&\leq \min\left\{(1+\epsilon+\frac{\lambda_0-s}{4s})\eta_0, \frac{N}{s}(\eta_0-\frac{Q_0}{\epsilon_1})\right\},
	\end{split}
	\end{align}
	where  $\epsilon_1=\frac{\epsilon}{1+2\epsilon}>0$, and 
	\begin{align}\label{cap_Q}
	Q_0:=\frac{N-s}{N}(b_{w}+b_{v}+\norm{A}p_0)+b_{w}.
	\end{align}
	By the first inequality and second inequality of \eqref{pf_beta}, we have $k_0^*=\min\{1,\frac{\beta}{\norm{A}(p_0+\eta_0)+b_{w}+b_{v}}\}=1$ and $\frac{s\beta}{N\eta_0}<1$, respectively. Then 
	\begin{align}\label{eq_mo}
	m_0:&=\left(1-\frac{s\beta}{N\eta_0}\right)\left(1-\frac{k_0^*\lambda_0}{N}\right)^{-1}\nonumber\\
	&=\left(\frac{N}{s}-\frac{\beta}{\eta_0}\right)\left(\frac{N-\lambda_0}{s}\right)^{-1}\nonumber\\
	&=1+\left(\frac{\lambda_0}{s}-\frac{\beta}{\eta_0}\right)\frac{s}{N-\lambda_0}\nonumber\\
	&\overset{(a)}{\geq} 1+\left(\frac{\lambda_0}{s}-1-\frac{\lambda_0-s}{4s}-\epsilon\right) \frac{s}{N-\lambda_0}\nonumber\\
	&=1+\frac{3(\lambda_0-s)}{4(N-\lambda_0)}- \frac{\epsilon s}{N-\lambda_0}\nonumber\\
	&\overset{(b)}{\geq} 1+2\epsilon
	\end{align}
	where $(a)$ is obtained by applying the second inequality of \eqref{pf_beta}, and 
	$(b)$ is derived by using $\lambda_0-s=4\epsilon(N-\lambda_0)$ and   $s\leq N-\lambda_0$.
	From the second inequality of \eqref{pf_beta}, we obtain 
	\begin{align}\label{pf_thm1_mid}
	\vartheta_0:=1-\frac{Q_0}{\eta_0}\left(1-\frac{s\beta}{N\eta_0}\right)^{-1}\geq 1-\epsilon_1=\frac{1+\epsilon}{1+2\epsilon}.
	\end{align} 
	By \eqref{eq_mo} and \eqref{pf_thm1_mid}, we have $\vartheta_0m_0\geq 1+\epsilon\geq \norm{A}$. It is easy to check that $\vartheta_0m_0\geq \norm{A}$ is equivalent    equation \eqref{condition_thm}. Thus, the sufficiency is satisfied in this case  with the above parameters, i.e., $\epsilon=\frac{\lambda_0-s}{4(N-\lambda_0)}$, and $\beta,\eta_0$ and $L>\frac{\ln \norm{A}}{\ln \gamma^{-1}}$ satisfying \eqref{pf_beta}.

	Case 2: For the attack-free case, i.e., $s=0$, we consider $\norm{A}\in [1,1+\epsilon)$ with  $\epsilon=\frac{\lambda_0}{4N-\lambda_0}>0$.  
	Similar to case 1, it is feasible to choose 
	sufficiently large $\beta,\eta_0$ and $L>\frac{\ln \norm{A}}{\ln \gamma^{-1}}$, such that 
	\begin{align}\label{pf_beta22}
	\begin{split}
	2\beta&\geq (1+\epsilon)(p_0+\eta_0)+b_{w}+b_{v},\\
	\frac{q_0}{\eta_0}&\leq \frac{\lambda_0}{4N}(1+\epsilon).
	\end{split}
	\end{align}
	From the first inequality of \eqref{pf_beta22}, we see  $k_0^*=\min\{1,\frac{\beta}{\norm{A}(p_0+\eta_0)+b_{w}+b_{v}}\}\geq  \frac{1}{2}$.
	With $k_0^*\geq \frac{1}{2}$,
	it is easy to check that   equation \eqref{condition_thm} is satisfied if 
	$\frac{1}{\norm{A}}\frac{\eta_0-q_0}{\eta_0}\geq 1-\frac{\lambda_0}{2N}$. This inequality is satisfied due to $\norm{A}\in [1,1+\epsilon)$ and the second inequality of \eqref{pf_beta22}. Thus, the sufficiency is satisfied in this case  with   $\epsilon=\frac{\lambda_0}{4N-\lambda_0}$, and $\beta,\eta_0$ and $L>\frac{\ln \norm{A}}{\ln \gamma^{-1}}$ satisfying \eqref{pf_beta22}.
	
	}

Note that for $\eta_0$ in the two cases above, we are able to make it bigger  such that the initial error condition in Assumption~
\ref{ass_bounds} holds.

2) Necessity:
We use the contradiction method. If  $\lambda_0>s$ does not hold,  i.e., $\lambda_0\leq s$. Equation \eqref{condition_thm} is equivalent to 
\begin{align}\label{eq_a2}
\frac{\eta_0-q_0}{\eta_0}\left(1-\frac{k^*_0\lambda_0}{N}\right)^{-1}\geq \norm{A}\geq 1.
\end{align}
If $\frac{s\beta}{N\eta_0}\geq 1$, from the form of $q_0$ in \eqref{notation}, we have  $\eta_0<q_0$. Then the left-hand side of \eqref{eq_a2} is negative, which contracts with the right-hand side of \eqref{eq_a2}. Thus, $\frac{s\beta}{N\eta_0}<1.$ With the same notations as  case 1 of the sufficiency proof, \eqref{eq_a2} is equivalent to  $\vartheta_0m_0\geq \norm{A}\geq 1$. Due to $\frac{s\beta}{N\eta_0}< 1$, we have $\vartheta_0<1.$  Then 	$m_0$ has to be larger than 1, which leads to $\frac{s\beta}{N\eta_0}<\frac{k^*_0\lambda_0}{N}$. It is equivalent to $\frac{s}{\lambda_0}<\frac{k^*_0\eta_0}{\beta}$. Due to $\frac{s}{\lambda_0}\geq 1$, we have $\frac{k^*_0\eta_0}{\beta}>1$, which however can not be satisfied due to $k_0^*=\min\{1,\frac{\beta}{\norm{A}(p_0+\eta_0)+b_{w}+b_{v}}\}$. Thus, the  conjecture  $\lambda_0\leq s$  is not right, which means $\lambda_0> s$.

\subsection{Proof of Theorem~\ref{thm_normal}}\label{app_thm_normal}

First, we prove 1).	 For $i\in\mathcal{A}^c(t)$,
by 2) of Theorem~\ref{thm_stability}, we have $\sup\limits_{t\geq  t_0}\norm{e_{i}(t)}\leq \rho_{ t_0}+\sup_{t\geq t_0}p(t)$, 
thus,
\begin{align*}
&\sup\limits_{t\geq  t_0+1}|y_{i}(t)-C_{i}A\hat x_{i}(t-1)|\\
\leq& \norm{A} \sup\limits_{t\geq  t_0+1}\norm{e_{i}(t-1)}+b_w+b_v\\
\leq &\norm{A}(\rho_{ t_0}+\sup_{t\geq t_0}p(t))+b_w+b_v.
\end{align*}	
If \eqref{extra_condition} holds,  by Algorithm~\ref{alg:A}, all the observations of the attack-free sensors will eventually not be saturated, i.e., $k_i(t)=1,\forall i\in\mathcal{A}^c(t),t\geq  t_0+1.$

Next, we prove 2). By 1) of this theorem, for $i\in\mathcal{A}^c(t)$,  we have  $k_i(t)=1$, $\forall t> t_0$, then $\norm{M_t}\leq \varpi,$ where $M_t$ is defined in \eqref{eq_notations}. According to the error dynamics in \eqref{eq_error_track_simple} and inequalities \eqref{eq_attack_error}--\eqref{eq_condition}, the  upper bound of $\norm{\tilde e(t)}$ is obtained  by  applying   Lemma \ref{lem_stability}. It follows from \eqref{eq_transition} that the bound is tighter than the one in  1) of Theorem~\ref{thm_stability}.
 Due to $\norm{e_i(t)}\leq \norm{\tilde e(t)}+\norm{\bar e_i(t)}$ and \eqref{eq_error_norm}, the upper bound of $\norm{e_i(t)}$ is obtained.

Finally, we prove 3). By the real-time upper bound of the estimation error, it is straightforward to have its limit superior. Next, we prove the limit superior bound is no larger than the one in 3) of Theorem~\ref{thm_stability}, i.e., $\frac{q_0}{1-\varpi}\leq \inf\limits_{t_0\in \Gamma} \rho_{ t_0}.$ Employing the properties $\inf x+y\geq \inf x+\inf y$ and $\inf xy\geq \inf x \inf y$ for $x,y>0$  on $F(\rho_{t_0})\rho_{t_0}+q_0=\rho_{t_0+1}\leq \rho_{t_0}$ yields
\begin{align*}
\inf\limits_{t_0\in \Gamma} \rho_{t_0}&\geq \frac{q_0}{1-\inf\limits_{t_0\in \Gamma} F(\rho_{t_0})}\\
&\overset{(a)}{\geq} \frac{q_0}{1-\norm{A}\left(1-\frac{1}{N}\lambda_0\right)}\\
&\geq\frac{q_0}{1-\varpi},
\end{align*} 
where $(a)$ holds by considering the expression of $F(\cdot)$ in \eqref{notation}.

\subsection{Proof of Proposition \ref{prop_resi}}\label{pf_prop_resi}
First, we consider the case of  $\norm{A}<1$.    By applying 1) of Theorem~\ref{thm_stability}  and   choosing $t_0=1$  and  $\bar q_0=b_{w}+\max\{\beta,b_{w}+b_{v}+\norm{A}p_0\}$, we have \eqref{eq_bound_mono}. From \eqref{eq_lambda} and \eqref{notation}, we see that $F(\eta_0)$ is a monotonically non-decreasing function w.r.t.  $s$. Thus $f(s)$ is a monotonically non-decreasing function w.r.t.  $s$.

Second, we consider the case of  $\norm{A}\geq 1$. In the case, we  have \eqref{eq_bound_mono}, by applying 1) of Theorem~\ref{thm_stability}, and by choosing $t_0=1$  and $\bar q_0=q_0$. Next, we show the $f(s)$ is a monotonically non-decreasing function w.r.t. $s$.  As discussed above that $F(\eta_0)$ is a monotonically non-decreasing function w.r.t.  $s$, we  just need to prove that $q_0$ is a monotonically non-decreasing function w.r.t.  $s$. This is obviously ensured if $\beta>b_{w}+b_{v}+\norm{A}p_0$. 
Next, we prove this point by contradiction. In other words, we assume $\beta\leq b_{w}+b_{v}+\norm{A}p_0$.   Note that \eqref{condition_thm} is equivalent to
\begin{align}\label{pf_1}
1-k^*(\eta_0) \frac{\lambda_0}{N}\leq \frac{1}{\norm{A}}\left(1-\frac{q_0}{\eta_0}\right).
\end{align}
Due to $\beta\leq b_{w}+b_{v}+\norm{A}p_0$, we have $q_0\geq \beta+ b_{w}$. Then a necessary to ensure \eqref{pf_1} is 
\begin{align}\label{pf_2}
1-k^*(\eta_0) \frac{\lambda_0}{N}\leq \frac{1}{\norm{A}}\left(1-\frac{\beta+ b_{w}}{\eta_0}\right).
\end{align}
It follows from   \eqref{notation} that   $k^*(\eta_0)=\frac{\beta}{\norm{A}(p_0+\eta_0)+b_{w}+b_{v}}$.
By substituting $k^*(\eta_0)$ into \eqref{pf_2}, we obtain
\begin{align*}
\frac{\beta}{\norm{A}(p_0+\eta_0)+b_{w}+b_{v}} \frac{\lambda_0}{N}\geq \frac{\beta+ b_{w}+(\norm{A}-1)\eta_0}{\norm{A}\eta_0},
\end{align*}
which can not be satisfied due to $\lambda_0\leq N$ and $1\leq\norm{A}.$ Therefore, the assumption $\beta\leq b_{w}+b_{v}+\norm{A}p_0$ does not hold.

\subsection{Proof of Theorem~\ref{thm_bounds_detected}}\label{pf_thm_bounds_detected}
The proof is similar to the proofs of Theorems~\ref{thm_stability}--\ref{thm_normal}. In the following, we just show the main points of this proof.

Given a time $T>0$ and the maximal  number of the detected sensors at time $t$, i.e., $d(T)$, similar to the proof of Theorem~\ref{thm_stability},  for $\forall t\geq T$, we construct the 
following sequence $\{\bar \rho_{t}\in\mathbb{R}|\bar \rho_{t}\}$ in \eqref{sequence2}.
It is straightforward to prove that for $\forall t\geq T$, $\norm{\tilde e(t)}\leq\bar \rho_{t}$, where $\tilde e(t)=\frac{1}{N}\sum_{i=1}^N\hat x_{i}(t)-x(t)$. Next, we study the relationship between $\bar \rho_{t}$ in \eqref{sequence2} and $\rho_{t}$ in \eqref{sequence}. Due to $\bar\rho_T=\rho_T$, we have $\bar \rho_{T+1}=  \rho_{T+1}-\frac{d(T)\beta}{N}.$ Then, for $t=T+2,$ we have 
\begin{align*}
\bar \rho_{T+2}&= F(\bar \rho_{T+1}) \rho_{T+1}+ q_0-(F(\bar \rho_{T+1})+1)\frac{d(T)\beta}{N}\\
&\leq  \rho_{T+2}-(F_{*}+1)\frac{d(T)\beta}{N}.
\end{align*}
where $F_{*}=\inf\limits_{t_0\in \bar\Gamma}F( \bar\rho_{ t_0})\in [0,1)$, and $\bar \Gamma=\{t\geq T|\bar\rho_{t}\leq \bar\rho_{t-1}\}$. 
By recursively applying the above operation, for $t\geq T$, we obtain
\begin{align*}
\bar \rho_{t}\leq \rho_{t}-\frac{d(T)\beta}{N}\left(\frac{1-F_{*}^{t-T}}{1-F_{*}}\right).
\end{align*}
Then by Lemma \ref{lem_stability} and Theorem~\ref{thm_stability}, we have $\limsup\limits_{t\rightarrow \infty}\norm{\tilde e(t)}\leq \inf\limits_{t_0\in \Gamma} \rho_{ t_0}-\frac{d(T)\beta}{N(1-F_{*})}.$ 
Thus, the first conclusion holds  by  applying   Lemma \ref{lem_stability}, $\norm{e_i(t)}\leq \norm{\tilde e(t)}+\norm{\bar e_i(t)}$, and \eqref{eq_error_norm}. 

\subsection{Proof of Theorem~\ref{thm_observer}}\label{pf_thm_observer}
Under condition 1), owing to the connectivity of   network $\mathcal{G}$, there is a common time $\tilde t\geq \hat t_0$, such that $d_j(\tilde t)=s,\forall j\in\mathcal{V}$.	
%
%
%
%
For $t\geq \tilde t$, all the observations of the attacked sensors are discarded. Then, we have the compact form of recursive state estimates of Algorithm~\ref{alg:B} in the following 
\begin{align}\label{eq_recursive2}
\hat X(t)=&\left(I_{Nn}-\alpha(\mathcal{L}\otimes I_n)\right)^L\bigg[(I_N\otimes A)\hat X(t-1)\nonumber\\
&+ \bar C^{\sf T} \bar K_{\mathcal{A}^c}(Y(t)-\bar C(I_N\otimes A)\hat X(t-1))\bigg].
\end{align}
Let $\bar E(t)=\hat X(t)-\textbf{1}_N\otimes \hat x_{avg}(t)$, i.e., $\bar E(t)=[\bar e_1^{\sf T}(t),\dots,\bar e_N^{\sf T}(t)]^{\sf T}.$
By referring to \cite{He2020Secured}, we have
\begin{align}\label{eq_bound}
&\norm{\bar E(t+1)}\nonumber\\
\leq &\norm{(I_N\otimes A)}\norm{\left(I_{Nn}-\alpha(\mathcal{L}\otimes I_n)-P_{Nn}\right)^L\bar E(t)}\nonumber\\
&+\big\|(I_{Nn}-P_{Nn}) \left(I_{Nn}-\alpha(\mathcal{L}\otimes I_n)\right)^L\nonumber\\
&\qquad\bar C^{\sf T}\bar  K_{\mathcal{A}^c}(Y(t)-\bar C(I_N\otimes A)\hat X(t-1))\big\|\nonumber\\
\leq &2\norm{A}\gamma^L\norm{\bar E(t)}+\norm{A}\gamma^L\sqrt{N-s}\norm{\tilde e(t)}.
\end{align}
Similar to \eqref{eq_error_track_simple}, for $i\in\mathcal{A}^c$, $t\geq \tilde t$,  $k_i(t)=1$, we have $\tilde e(t+1)=M_{21}\tilde e(t)-M_{22}\bar E(t),$
	where  $M_{21}=\left(I_n-\frac{1}{N}\sum_{i\in \mathcal{A}^c}C_{i}^{\sf T}C_{i}\right)A$ and $M_{22}=\frac{1}{N}(\textbf{1}_N^{\sf T}\otimes  I_n)\bar C^{\sf T}\bar K_{\mathcal{A}^c}\bar C(I_N\otimes A).$
	Then it holds that
	\begin{align}\label{eq_bound2}
	\norm{\tilde e(t+1)}\leq \tau_0\norm{\bar E(t)}+\varpi\norm{\tilde e(t)},
	\end{align}
where $\varpi$ is given in Theorem~\ref{thm_normal}.
By \eqref{eq_bound} and \eqref{eq_bound2}, if the matrix $\left(\begin{smallmatrix}
2\norm{A}\gamma^L&\norm{A}\gamma^L\sqrt{N-s}\\
\tau_0&\varpi
\end{smallmatrix}\right)$
is Schur stable,   $\norm{\tilde e(t)}$ and $\norm{\bar E(t)}$ go to zero asymptotically. Thus, $\norm{e_i(t)}$ is convergent to zero as time goes to infinity.

%

%
%
%
%
\small
\bibliography{All_references}
\bibliographystyle{ieeetr}

\end{document}